\documentclass[10pt,a4paper,manyauthors]{fundam}
\usepackage[T1]{fontenc}
\usepackage[english]{babel}
\usepackage{amsmath}
\usepackage{amsfonts}
\usepackage{amssymb} 
\usepackage{mathrsfs}
\usepackage{mathtools}
\usepackage{indentfirst}
\usepackage{fancyhdr}
\usepackage{color}
\usepackage{tikz}
\usepackage{listings}
\usepackage{pgfplots}
\usepackage{xcolor}
\usepackage{sidecap}
\usepackage{pgf}
\usetikzlibrary{arrows,automata,positioning,fit,calc,shapes}
\usepackage{graphicx}
\usepackage{empheq}
\usepackage{caption}
\usepackage[babel]{csquotes}
\usepackage[ruled,vlined]{algorithm2e}
\usepackage{subfigure}
\usepackage{epsfig}

\newtheorem{conjecture}{Conjecture}
\begin{document}

\setcounter{page}{1}
\issue{XXI~(2020)}

\title{A linear bound on the k-rendezvous time for primitive sets of NZ matrices}

\author{Costanza Catalano\corresponding\\
Department of Economics, Statistics and Research\\
Banca d'Italia (Central Bank of Italy)\\
Largo Guido Carli 1, 00044 Frascati (Roma), Italy.\\
costanzacatalano{@}gmail.com\\
$  $\\
$  $\\
\and Umer Azfar\\
ICTEAM,
Universit\'{e} Catholique de Louvain\\
Avenue Georges Lema\^{i}tres 4-6, Louvain-la-Neuve, Belgium.\\
umer.azfar{@}student.uclouvain.be\\
\and Ludovic Charlier\\
ICTEAM,
Universit\'{e} Catholique de Louvain\\
Avenue Georges Lema\^{i}tres 4-6, Louvain-la-Neuve, Belgium.\\
ludovic.charlier{@}student.uclouvain.be\\
\and Rapha\"{e}l M. Jungers\thanks{R. M. Jungers is a FNRS Research Associate. He is supported by the French Community of Belgium, the Walloon Region and the Innoviris Foundation.}\\
ICTEAM,
Universit\'{e} Catholique de Louvain\\
Avenue Georges Lema\^{i}tres 4-6, Louvain-la-Neuve, Belgium.\\
raphael.jungers{@}uclouvain.be
} 
\maketitle

\runninghead{C. Catalano, U. Azfar, L. Charlier, R.M. Jungers}{A linear bound on the k-rt for primitive sets of NZ matrices}

\begin{abstract}
A set of nonnegative matrices is called primitive if there exists a product of these matrices that is entrywise positive. Motivated by recent results relating synchronizing automata and primitive sets, we study the length of the shortest product of a primitive set having a column or a row with $ k $ positive entries, called its \emph{$ k $-rendezvous time}  (\emph{$ k$-RT}), in the case of sets of matrices having no zero rows and no zero columns. We prove that the $ k $-RT is at most linear w.r.t.\ the matrix size $ n $ for small $ k $, while the problem is still open for synchronizing automata. We provide two upper bounds on the $ k $-RT: the second is an improvement of the first one, although the latter can be written in closed form. We then report numerical results comparing our upper bounds on the $ k $-RT with heuristic approximation methods.
\end{abstract}


\section{Introduction}
\textbf{Primitive sets of matrices.}
A nonnegative matrix $ M $ is called \emph{primitive} if there exists an integer $ s\in\mathbb{N} $ such that $ M^s>0 $ entrywise. This notion was introduced by Perron and Frobenius at the beginning of the 20th century, and it can be extended to \emph{sets} of matrices: a set of nonnegative matrices $ \mathcal{M}=\lbrace M_1,\dots ,M_m\rbrace $  is called \emph{primitive} if there exist some indices $ i_1,\dots ,i_r\in \lbrace 1,\dots ,m\rbrace $ such that the product $ M_{i_1}\cdots M_{i_r} $ is entrywise positive. A product of this kind is called a \emph{positive} product and the length of the shortest positive product of a primitive set $ \mathcal{M} $ is called its \emph{exponent} and it is denoted by $ exp(\mathcal{M}) $. The concept of primitive set has been just recently formalized by Protasov and Voynov \cite{ProtVoyn}, but it had appeared before in different fields as in stochastic switching systems \cite{hennion1997,Protasov2011} and time-inhomogeneous Markov chains \cite{Hart,seneta}. It has lately gained more importance due to its applications in consensus of discrete-time multi-agent systems \cite{Pierre}, cryptography \cite{Fomichev2018} and automata theory \cite{BlonJung,GerenGusJung,Catalano,CatalanoJALC}. 
Deciding whether a set is primitive is a PSPACE-complete problem  \cite{GerenGusJung}, while computing the exponent of a primitive set is an FP$^{\text{NP}[\log]}$-complete problem \cite{GerenGusJung}; for the complexity of other problems related to primitivity and the computation of the exponent, we refer the reader to \cite{GerenGusJung}. For sets of matrices having at least one positive entry in every row and every column (called \emph{NZ} \cite{GerenGusJung} or \emph{allowable} matrices \cite{hajnal_1976,hennion1997}), the primitivity problem becomes decidable in polynomial-time \cite{ProtVoyn}, although computing the exponent remains NP-hard \cite{GerenGusJung}. Methods for approximating the exponent have been proposed \cite{CatalanoSPF,CatalanoSPF2} as well as cubic upper bounds on the matrix size $ n $ \cite{BlonJung}. 
Better upper bounds have been found for some classes of primitive sets (see e.g.\ \cite{GerenGusJung} and \cite{Hart}, Corollary 2.5).
The NZ condition is often met in applications and in particular in the connection with synchronizing automata.

\textbf{Synchronizing automata.}
A \emph{(complete deterministic finite state) automaton} is a $ 3 $-tuple $\mathcal{A}= \langle Q,\Sigma,\delta\rangle $ where $ Q=\lbrace q_1,\dots, q_n\rbrace $ is a finite set of states, $ \Sigma=\lbrace a_1,\dots ,a_m\rbrace $ is a finite set of input symbols (the \emph{letters} of the automaton) and $ \delta: Q\times\Sigma\rightarrow Q $ is the \emph{transition function}. Given some indices $i_1, i_2,...,i_l \in \{1,...,m\}$, we call $w = a_{i_1}a_{i_2}...a_{i_l}$ a \emph{word} and we define $\delta(q,w) = \delta(\delta(q,a_{i_1}a_{i_2}...a_{i_{l-1}}),a_{i_l})$.
An automaton is  \emph{synchronizing} if it admits a word $ w $, called a \textit{synchronizing} or a \emph{reset} word, and a state $ q $ such that $ \delta(q',w)=q $ for any state $ q'\in Q $. In other words, the reset word $ w $ brings the automaton from every state to the same fixed state. 

\begin{remark}\label{rem:aut}
The automaton $ \mathcal{A} $ can be equivalently represented by the set of matrices $ \lbrace A_1,\dots ,A_m\rbrace $ where, for all $ i=1,\dots ,m $ and $ l,k=1,\dots ,n $, $ (A_i)_{lk}=1 $ if $ \delta(q_l,a_i)=q_k $, $ (A_i)_{lk}=0 $ otherwise. The action of a letter $ a_i $ on a state $ q_j $ is represented by the product $ e_j^{\top}A_i $, where $ e_j $ is the $ j $-th element of the canonical basis.
Notice that the matrices $ A_1,\dots ,A_m $ are binary\footnote{A matrix is \emph{binary} if it has entries in $ \lbrace 0,1\rbrace $.} and row-stochastic, i.e.\ each of them has exactly one entry equal to $ 1 $ in every row and zero everywhere else. In this representation, the automaton $ \mathcal{A} $ is synchronizing if and only if there exists a product of its matrices with a column with all entries equal to $ 1 $ (also called an \emph{all-ones} column).
\end{remark}
The idea of synchronization is quite simple: we want to restore control over a device whose current state is unknown. For this reason, synchronizing automata are often used as models of error-resistant systems \cite{Epp,Chen}, but they also find application in other fields such as in symbolic dynamics \cite{mateescu}, in robotics \cite{Natarajan} or in resilience of data compression \cite{SCHU,Biskup}. For a recent survey on synchronizing automata we refer the reader to \cite{Volk}. We are usually interested in the length of the shortest reset word of a synchronizing automaton $ \mathcal{A} $, called its \emph{reset threshold} and denoted by $ rt(\mathcal{A}) $. Despite the fact that determining whether an automaton is synchronizing can be done in polynomial time (see e.g.\ \cite{Volk}), computing its reset threshold is an NP-hard problem \cite{Epp}\footnote{Moreover, even approximating the reset threshold of an $n$-state synchronizing automaton within a factor of $n^{1-\epsilon}$ is known to be NP-hard for any $\epsilon > 0$, see \cite{StrongInapprox}.}. One of the most longstanding open questions in automata theory concerns the maximal reset threshold of a synchronizing automaton, problem that is traditionally known as \emph{The \v{C}ern\'{y} conjecture}\footnote{\v{C}ern\'{y}, together with Pirick\'{a} and Rosenaurov\'{a}, explicitly stated it in 1971 \cite{CernyPiricka}, while the first printed version of such conjecture is attributable to Starke in 1966 \cite{Starke} (see also its recent english translation \cite{StarkeEN}). For further details on the paternity of this conjecture, we refer the reader to \cite{VolkovINTRO}.}.

\begin{conjecture}[The \v{C}ern\'{y}(-Starke) conjecture]
Any synchronizing automaton on $ n $ states has a synchronizing word of length at most $ (n-1)^2 $.
\end{conjecture}

\v{C}ern\'{y} also presented in his pioneering paper \cite{Cerny} (see also its recent english translation \cite{CernyEN}) a family of automata having reset threshold of exactly $ (n-1)^2 $, thus demonstrating that the bound in his conjecture (if true) cannot be improved. Exhaustive search confirmed the \v{C}ern\'{y} conjecture for small values of $ n $ \cite{Ananichev2016,BondtDon,Experiments,Traht06} and within certain classes of automata (see e.g. \cite{Kari,Steinberg,Volkov2007}), but despite a great effort has been made to prove (or disprove) it in the last decades, its validity still remains unclear. Indeed on the one hand, the best upper bound known on the reset threshold of any synchronizing $ n $-state automaton is cubic in $n $ \cite{Pin,Frankl,Szykula}, while on the other hand automata having quadratic reset threshold, called \emph{extremal} automata, are very difficult to find and few of them are known (see e.g.\ \cite{Babai,GusevSzikulaDzyga,Rystsov1997,Szykula2015}). Some of these families have been found by Ananichev et.\ al.\  \cite{SlowAutom} by coloring the digraph of primitive matrices having large exponent; this has been probably the first time where primitivity has been succesfully used to shed light on synchronization.

\textbf{Connecting primitivity and synchronization.}
The following definition and theorem establish the connection between primitive sets of binary NZ matrices and synchronizing automata. From here on, we will use the matrix representation of deterministic finite automata as described in Remark \ref{rem:aut}. 
\begin{definition}\label{def:assoc_autom}
 Let $ \mathcal{M}$ be a set of binary NZ matrices. The \emph{automaton associated to} the set $ \mathcal{M} $ is the automaton $  Aut(\mathcal{M}) $ such that $A\in Aut(\mathcal{M})$ if and only if $ A $ is a binary and row-stochastic matrix and there exists $ M\in \mathcal{M} $ such that $ A\leq M $ (entrywise). We denote with $  Aut(\mathcal{M}^{\top}) $ the automaton associated to the set $ \mathcal{M}^{\top}=\lbrace M^{\top}_1,\dots ,M^{\top}_m\rbrace  $.
\end{definition}

The following example exhibits a primitive set $ \mathcal{M}$ of NZ matrices and the synchronizing automata $  Aut(\mathcal{M}) $ and $ Aut(\mathcal{M}^{\top}) $.
\begin{example}\label{ex}
Here we present a primitive set and its associated automata, see also Figure \ref{fig:twoautom2}.
\begin{align*}
\mathcal{M}\!&=\!\left\lbrace  
\left( \begin{smallmatrix} 0 & 1&0 \\ 0&0&1 \\  1&0&0 \end{smallmatrix}\right) , \left( \begin{smallmatrix} 0 & 1&0 \\  1&0&1 \\ 0&0 & 1 \end{smallmatrix}\right) \right\rbrace , \\
Aut(\mathcal{M})\!&=\!\Bigl\lbrace  
a=\left( \begin{smallmatrix} 0 & 1&0 \\ 0&0&1 \\  1&0&0 \end{smallmatrix}\right),\, b_1 =\left( \begin{smallmatrix} 0 & 1&0 \\  1&0&0 \\ 0&0 & 1 \end{smallmatrix}\right) ,\,b_2=\left(  \begin{smallmatrix} 0 & 1&0 \\  0&0&1 \\ 0&0 & 1 \end{smallmatrix}\right) \Bigr\rbrace\\
Aut(\mathcal{M}^{\top})\!&=\!\Bigl\lbrace  
a'=\left( \begin{smallmatrix} 0 & 0&1 \\ 1&0&0 \\  0&1&0 \end{smallmatrix}\right) , \,b_1=\left(  \begin{smallmatrix} 0 & 1&0 \\  1&0&0 \\ 0&0 & 1 \end{smallmatrix}\right),\, b'_2=\left(  \begin{smallmatrix} 0 & 1&0 \\  1&0&0 \\ 0&1 & 0 \end{smallmatrix}\right) \Bigr\rbrace . 
\end{align*}
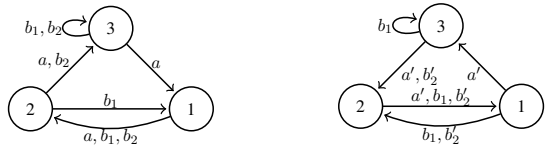
\begin{figure}
\centering
\begin{tikzpicture}[shorten >=1pt,node distance=2.5cm,on grid,auto,scale=0.6,transform shape,inner sep=0pt,bend angle=15,line width=0.2mm]
\node[state]    (q_3) {3 };
 			\node[state]    (q_1) [ below right=of q_3] {1};
 			\node[state]          (q_2) [below left=of q_3] {2};

 			\path[->]
 			(q_1) edge	[bend left=20]	 node  {$ a,b_1,b_2 $} (q_2)
 			(q_2) edge []	node  {$ b_1 $} (q_1)
 			(q_3) edge [loop left]		node  {$ b_1,b_2 $ } ()
 			(q_2) edge []		node  {$ a,b_2 $} (q_3)
 			(q_3) edge []		node  {$ a $} (q_1)
 			;  
\end{tikzpicture}$ \qquad \qquad$
\begin{tikzpicture}[shorten >=1pt,node distance=2.5cm,on grid,auto,scale=0.6,transform shape,inner sep=0pt,bend angle=15,line width=0.2mm]


\node[state]    (q_3) {3 };
 			\node[state]    (q_1) [ below right=of q_3] {1};
 			\node[state]          (q_2) [below left=of q_3] {2};

 			\path[->]
 			(q_1) edge	[bend left=20]	 node  {$ b_1,b'_2 $} (q_2)
 			(q_2) edge []	node  {$ a',b_1,b'_2 $} (q_1)
 			(q_3) edge [loop left]		node  {$ b_1 $ } ()
 			(q_3) edge []		node  {$ a',b'_2 $} (q_2)
 			(q_1) edge []		node  {$ a' $} (q_3)
 			;  
\end{tikzpicture}
\caption{The automata $ Aut(\mathcal{M}) $ (left) and $ Aut(\mathcal{M}^{\top}) $ (right) of Example \ref{ex}.}\label{fig:twoautom2}
\end{figure}
\end{example}
The following theorem establishes how $ exp(\mathcal{M}) $, $ rt\bigl( Aut(\mathcal{M})\bigr) $ and $ rt\bigl(Aut(\mathcal{M}^{\top})\bigr) $ are in relation to each others.

\begin{theorem}[\cite{BlonJung} Theorems 16-17, \cite{GerenGusJung} Theorem 2]\label{thm:autom_matrix}
Let $ \mathcal{M}\!=\!\lbrace M_1,\dots ,M_m\rbrace $ be a primitive set of $ n\times n $ binary NZ matrices. Then $  Aut(\mathcal{M})$ and $ Aut(\mathcal{M}^{\top})$ are synchronizing and it holds that: 
 
\begin{equation}\label{eq:thmauotm_mat}
rt\bigl( Aut(\mathcal{M})\bigr)\leq exp(\mathcal{M}) \leq rt\bigl( Aut(\mathcal{M})\bigr)+rt\bigl(Aut(\mathcal{M}^{\top})\bigr)+n-1.
\end{equation}

\end{theorem}

\begin{example}
Consider the set $ \mathcal{M} $ and the automata $ Aut(\mathcal{M}) $ and $ Aut(\mathcal{M}^{\top}) $ of Example \ref{ex}. It holds that $ exp(\mathcal{M})\!=\!7 $, $ rt\bigl( Aut(\mathcal{M})\bigr)\!=\!2 $ and $ rt\bigl(Aut(\mathcal{M}^{\top})\bigr)\!=\!3 $, thus showing that the upper bound of Eq.(\ref{eq:thmauotm_mat}) is tight. A reset word for $ Aut(\mathcal{M}) $ is $ w=b_2b_2 $, while a reset word for $Aut(\mathcal{M}^{\top})$ is $ w=b'_2a'b'_2 $.
\end{example}

Notice that the requirement in Theorem \ref{thm:autom_matrix} that the set $\mathcal{M} $ has to be made of \emph{binary} matrices is not restrictive, as the primitivity property does not depend on the magnitude of the positive entries of the matrices of the set. We can thus restrict ourselves to the set of binary matrices by using the Boolean product between them\footnote{In other words, we work with matrices over the Boolean semiring.}; this means that for any $ A $ and $ B $ binary matrices, we set $ (AB)_{ij}=1 $ any time that $ \sum_{s}A_{is}B_{sj}>0 $.  In this framework, primitivity can be also rephrased as a \emph{membership problem} (see e.g.\ \cite{Paterson,Potapov2017}), where we ask whether the all-ones matrix belongs to the semigroup generated by the matrix set. 

Equation (\ref{eq:thmauotm_mat}) shows that the behavior of the exponent of a primitive set of NZ matrices is tightly connected to the behavior of the reset threshold of its associated automaton. A primitive set $ \mathcal{M} $ with quadratic exponent implies that one of the automata $  Aut(\mathcal{M}) $ or $ Aut(\mathcal{M}^{\top}) $ has quadratic reset threshold; in particular, a primitive set with exponent greater than $ 2(n-1)^2+n-1 $ would disprove the \v{C}ern\'{y} conjecture. 
This property has been used by the authors in \cite{CatalanoJALC} to construct a randomized procedure for finding extremal synchronizing automata.

The synchronization problem for automata is about finding the length of the shortest word mapping the whole set of $ n $ states onto one single state. We can weaken this request by asking what is the length of the shortest word mapping $ k $ states onto one single state, for $ 2\leq k\leq n $. In the matrix framework, we are asking what is the length of the shortest product having a column with $ k $ positive entries. The case $ k=2 $ is trivial, as any synchronizing automaton has a letter mapping two states onto one; for $ k=3 $ Gonze and Jungers \cite{Gonze2015} presented a quadratic upper bound in the number of the states of the automaton  while, to the best of our knowledge, the cases $ k\geq 4 $ are still open. Clearly, the case $ k=n $ is the problem of computing the reset threshold.

In view of the connection between synchronizing automata and primitive sets, we extend the above described problem to primitive sets by introducing the \emph{k-rendezvous time} ($ k $-RT): the $ k $-RT of a primitive set $ \mathcal{M} $ is the length of its shortest product having a row or a column with $ k $ positive entries. The following proposition shows how the $ k $-RT of a primitive set $ \mathcal{M} $ of NZ matrices (denoted by $ rt_k(\mathcal{M}) $) is linked to the length of the shortest word for which there exists a set of $ k $ states mapped by it onto a single state in the automata $ Aut(\mathcal{M}) $ and $ Aut(\mathcal{M}^{\top}) $, where the lengths are denoted respectively by $ rt_k(Aut(\mathcal{M})) $ and $ rt_k(Aut(\mathcal{M}^{\top})) $.
\begin{proposition}\label{prop:krt_prim_aut}
Let $ \mathcal{M} $ be a primitive set of $ n\times n $ binary NZ matrices and let $ Aut(\mathcal{M}) $ and $ Aut(\mathcal{M}^{\top}) $ be the automata defined in Definition \ref{def:assoc_autom}. Then for every $ 2\leq k\leq n $, it holds that
\[\,
rt_k(\mathcal{M})=\min \bigl\lbrace rt_k\bigl(Aut(\mathcal{M})\bigr), rt_k\bigl(Aut(\mathcal{M}^{\top})\bigr)\bigr\rbrace\, .
\]
\end{proposition}

The proof of Proposition \ref{prop:krt_prim_aut} mimics the proof of Theorem 16 in \cite{BlonJung}; for the sake of completeness, we provide a self-contained proof.

\begin{proof}
By Definition \ref{def:assoc_autom}, each matrix of $ Aut(\mathcal{M}) $ and $ Aut(\mathcal{M}^{\top}) $ is entrywise smaller than a matrix of $ \mathcal{M} $. It follows that $ rt_k(\mathcal{M})\leq\min \bigl\lbrace rt_k\bigl(Aut(\mathcal{M})\bigr), rt_k\bigl(Aut(\mathcal{M}^{\top})\bigr)\bigr\rbrace $. 
\\Let now $ M=M_{i_1}\cdots M_{i_u} $ be the product that attains the $ k $-RT, that is a product of length $ rt_k(\mathcal{M}) $ having a column or a row with $ k $ positive entries. Suppose that $ M $ has column with $ k $ positive entries: we show that $ rt_k(\mathcal{M})\geq rt_k\bigl(Aut(\mathcal{M})\bigr) $. Let $ j $ be the index of this column and $ S $ be its support.
We claim that for every $  r\in [u] $ we can safely set to zero some entries of $ M_{i_r} $ in order to make its rows be stochastic while making sure that the final product still has the $ j $-th column with support $ S $. In other words, we claim that for every $ r\in [u] $ we can select a binary row-stochastic matrix $ A_{r}\leq M_{i_r}$ (entrywise) such that the $ j $-th column of the product $ A_{1}\cdots A_{u}$ has support $ S $. If this is true, since by hypothesis the matrices $ A_{1},\dots ,A_{u} $ belong to $  Aut(\mathcal{M}) $, it holds that $ rt_k\bigl(Aut(\mathcal{M})\bigr)\leq rt_k(\mathcal{M}) $.\\
We now prove the claim: let $ D_r $ be the digraph on $ n $ vertices and edge set $ E_r $ such that $ p\rightarrow q\in E_r $ if and only if $ (M_{i_r})_{pq}>0 $. The fact that the $j $-th column of $ M_{i_1}\cdots M_{i_u}$ has support $ S $ means that for every $ s\in S $ there exists a sequence of vertices $ v^s_1,\dots ,v^s_{u+1}\in [n] $ such that:
\begin{align}
&v^s_1=s\enspace ,\label{eq:fin1new}\\
&v^s_{u+1}=j\enspace ,\label{eq:fin2new}\\
& v^s_r\rightarrow v^s_{r+1}\in E_r\quad \forall r=1,\dots ,u\enspace .\label{eq:fin3new}
\end{align}
We can impose an additional property on these sequences: if at step $ t $ two sequences share the same vertex, then they have to coincide for all the steps $ t'>t $. More formally, if for some $ t\in [u] $ we have that $ v^s_t=v^{s'}_t $ for $ s\neq s' $, then we set $  v^{s'}_{t'}=v^{s}_{t'} $ for all $ t'>t $ as the new sequence $ v^{s'}_1,\dots ,v^{s'}_t,v^{s}_{t+1},\dots ,v^{s}_{u+1}  $ for vertex $ s' $ fulfills all the requirements (\ref{eq:fin1new}), (\ref{eq:fin2new}) and (\ref{eq:fin3new}). For every $ r\in [u] $, we now remove from $ E_r $ all the edges that are not of type (\ref{eq:fin3new}). Furthermore, for every $ r\in [u] $ and vertex $ w\notin \lbrace v^s_r\rbrace_{s\in S} $, we remove from $ E_r $ all the outgoing edges of $ w $ but one. We call this new edge set $ \tilde{E}_r $ and let $ \tilde{D}_r $ be the subgraph of $ D_r $ with edge set $ \tilde{E}_r $. Then, for every $ r\in [u] $, we set $ A_r $ to be the adjacency matrix of $ \tilde{D}_r $. Since $M_{i_r}$ is NZ, if $ A_r $ has some zero-rows we can always add a one in each of them while preserving the property $A_r \le M_{i_r}$. We do so in order to make $A_r$ row stochastic. By construction, for all $ r\in [u] $, $ A_r $ has exactly one positive entry in each row and it is entrywise smaller than $ M_{i_r} $, so $ A_r\in Aut(\mathcal{M}) $. Finally, the $ j $-th column of  $ A_{1}\cdots A_{u} $ has support $ S $ by construction.\\
The case when $ M $ has a row with $ k $ positive entries can be proved via a similar reasoning by observing that the product $ M^{\top}=M^{\top}_{i_u}\cdots M^{\top}_{i_1} $ has a column with $ k $ positive entries, and so for every $ r\in [u] $ we can select a binary matrix $ B_{r}\leq M^{\top}_{i_r}$ (entrywise) such that $ B_r\in Aut(\mathcal{M}^{\top}) $ and $ B_{1}\cdots B_{u} $ has a column with $ k $ positive entries. This implies that $ rt_k\bigl(Aut(\mathcal{M}^{\top})\bigr)\leq rt_k(\mathcal{M})$.
\end{proof}

\textbf{Our contribution.}
This paper comes as an extended version of the one published at the conference \emph{Developments in Language Theory 2019} \cite{AzfarCatalano}. With respect to the conference version, the entire Sections \ref{sec:improv_ub} and \ref{sec:num_F} have been added; also other small parts and some notation have been changed.
In this work we provide an analytical upper bound on $rt_k(\mathcal{M})  $ that holds for any primitive set $ \mathcal{M} $ of $ n\times n $ NZ matrices. This upper bound is a function of $ n $ and $ k$, and it proves in particular that the $k$-rendezvous time $ rt_k(\mathcal{M}) $ is upper bounded by a linear function in $ n $ for any fixed $ k\leq\sqrt{n} $, problem that is still open for synchronizing automata. Our result also implies that for any fixed $ k\leq\sqrt{n} $,
$ \min \bigl\lbrace rt_k\bigl(Aut(\mathcal{M})\bigr), rt_k\bigl(Aut(\mathcal{M}^{\top})\bigr)\bigr\rbrace $ is upper bounded by a linear function in $ n $. This is presented in Section \ref{sec:krt}; in particular, in Subection \ref{sec:num_B} we report some numerical experiments and we show that this first technique for upper bounding $ rt_k(\mathcal{M}) $ cannot be much improved as it is. 
We then present in Section \ref{sec:improv_ub} a second upper bound for the $ k $-RT that improves the first one for any $ 2\leq k\leq n $. We report some numerical experiments in Section \ref{sec:num_F} comparing our two theoretical upper bounds on the $ k $-RT with the real one (or an approximation when it becomes too hard to compute) for some examples of primitive sets. Finally, as the second upper bound cannot be written in closed form, in the same section we present some graphs picturing its behavior, showing that when $ n $ is not too big with respect to $ k $ the second upper bound significantly improves on the first one.

\section{Notation and preliminaries}\label{sec:notation}
The set $ \lbrace 1,\dots ,n\rbrace $ is represented by $ [n] $. The \emph{support} of a nonnegative vector $ v $ is the set $supp(v)= \lbrace i: v_i>0\rbrace $ and the \emph{weight} of a nonnegative vector $ v $ is the cardinality of its support.

Given a matrix $ A $, we denote by $ A_{*j} $ its $ j $-th column and by $ A_{i*}$ its $ i $-th row. 
A \emph{permutation} matrix is a binary matrix having exactly one positive entry in every row and every column. We remind that an $n\times n$ matrix $ A $ is called \emph{irreducible} if for any $i,j \in [n]$, there exists a natural number $k$ such that $(A^k)_{ij}>0$. A matrix $A$ is called \emph{reducible} if it is not irreducible.
 
 Given $ \mathcal{M} $ a set of matrices, we denote by $ \mathcal{M}^d $ the set of all the products of at most $ d $ matrices from $ \mathcal{M}$. A set of matrices $ \mathcal{M}=\lbrace M_1,\dots ,M_m\rbrace $ is \emph{reducible} if the matrix $ \sum_{i}M_i $ is reducible, otherwise it is called \emph{irreducible}. 
 Irreducibility is a necessary but not sufficient condition for a matrix set to be primitive (see \cite{ProtVoyn}, Section 1). Given a directed graph $ D=(V,E) $, we denote by $ v\rightarrow w $ the directed edge leaving $ v $ and entering in $ w $ and by $ v\rightarrow w\in E $ the fact that the edge $ v\rightarrow w  $ belongs to the digraph $ D $. A directed graph is \emph{strongly connected} if there exists a directed path from any vertex to any other vertex.

\begin{lemma}\label{lem:prem_col}
Let $ \mathcal{M} $ be an irreducible set of $ n\times n $ NZ matrices, $ A\in \mathcal{M} $ and $ i,j\in [n] $. Then there exists a matrix $ B\in \mathcal{M}^{n-1} $ such that $ supp(A_{*i}) \subseteq supp((AB)_{*j}) $.
\end{lemma}

\begin{proof}
We consider the labeled directed graph $ \mathscr{D}_{\mathcal{M}}=(V, E) $ where $ V=[n] $ and $ i\rightarrow j\in E $ iff there exists a matrix $ A\in \mathcal{M} $ such that $ A_{ij}>0$. We label the edge $ i\rightarrow j\in E $ by all the matrices $ A\in \mathcal{M} $ such that $ A_{ij}>0$.  
Notice that a path in $ \mathscr{D}_{\mathcal{M}} $ from vertex $ k $ to vertex $ l $ having the edges sequentially labeled by the matrices $ A_{s_1},\dots ,A_{s_r} \in \mathcal{M} $ means that $ (A_{s_1}\cdots A_{s_r})_{kl}>0 $. Since $ \mathcal{M} $ is irreducible, it follows that $ \mathscr{D}_{\mathcal{M}} $ is strongly connected and so any pair of vertices in $ \mathscr{D}_{\mathcal{M}} $ is connected by a path of length at most $ n-1 $. Consider a path connecting vertex $ i $ to vertex $ j $ whose edges are sequentially labeled by the matrices $ A_{s_1},\dots ,A_{s_t} $ and let $ B=A_{s_1}\cdots A_{s_t} $. Clearly $ B\in \mathcal{M}^{n-1} $; furthermore it holds that $ B_{ij}>0 $ and so $ supp( A_{*i}) \subseteq supp\bigl( (AB)_{*j}\bigr) $.
\end{proof}
The following definition will be crucial for the results in the next sections.

\begin{definition}\label{defn:sg}
Let $ \mathcal{M} $ be an irreducible set of $ n\times n $ NZ matrices. We define the \emph{pair digraph} of the set $ \mathcal{M} $ as the labeled directed graph $ \mathcal{PD}(\mathcal{M})=(\mathcal{V},\mathcal{E} ) $ where $ \mathcal{V}=\lbrace (i,j): 1\leq i\leq j\leq n\rbrace $ and $ (i,j)\rightarrow (i',j') \in \mathcal{E} $ if and only if there exists $ A\!\in\! \mathcal{M} $ such that 
\begin{equation}\label{eq:sg}
 A_{ii'}>0  \text{ and }  A_{jj'}>0, \text{ or } A_{ij'}>0 \text{ and } A_{ji'}>0. 
\end{equation}
An edge $ (i,j)\!\rightarrow\! (i',j')\! \in\!\mathcal{E}$ is labeled by every matrix $ A\!\in\! \mathcal{M} $ for which Eq.\ (\ref{eq:sg}) holds. A vertex of type $ (s,s) $ is called a \emph{singleton}.
\end{definition}

\begin{lemma}\label{lem:sg}
Let $ \mathcal{M} $ be a finite set of $ n\times n $ NZ matrices and let $ \mathcal{PD}(\mathcal{M})=(\mathcal{V},\mathcal{E} ) $ be its pair digraph. Let $ i,j,k\in [n] $ and suppose that there exists a path in $ \mathcal{PD}(\mathcal{M})$ from the vertex $ (i,j) $ to the singleton $ (k,k) $ having the edges sequentially labeled by the matrices $ A_{s_1},\dots ,A_{s_l} \in \mathcal{M}$. Then it holds that for every $ A\in \mathcal{M} $, \[ supp(A_{*i})\cup supp( A_{*j}) \subseteq supp( (AA_{s_1}\cdots A_{s_l})_{*k}) \,. \] Suppose now that $ \mathcal{M} $ is irreducible. Then it holds that $ \mathcal{M} $ is primitive if and only if for any $ (i,j)\in\mathcal{V} $ there exists a path in $ \mathcal{PD}(\mathcal{M})$ from $ (i,j) $ to some singleton.
\end{lemma}

\begin{proof}
By the definition of the pair digraph $\mathcal{PD}(\mathcal{M}) $ (Definition \ref{defn:sg}), the existence of a path from vertex $ (i,j) $ to vertex $ (k,k) $ labeled by the matrices $ A_{s_1},\dots ,A_{s_l} $ implies that $ (A_{s_1}\cdots A_{s_l})_{ik}>0 $ and $ (A_{s_1}\cdots A_{s_l})_{jk}>0 $. By Lemma \ref{lem:prem_col}, it follows that  $ supp( A_{*i})\cup supp( A_{*j}) \subseteq supp\bigl( (AA_{s_1}\cdots A_{s_l})_{*k}\bigr)  $.\\
Suppose now that $ \mathcal{M} $ is irreducible. If $ \mathcal{M} $ is primitive, then there exists a product $M$ of matrices from $\mathcal{M}$ such that for all $i,j$,  $M_{ij}>0$. By the definition of $\mathcal{PD}(\mathcal{M})$, this implies that any vertex in $\mathcal{PD}(\mathcal{M})$ is connected to any other vertex. On the other hand, if every vertex in $ \mathcal{PD}(\mathcal{M})$  is connected to some singleton, then for every $ i,j,k\in [n] $ there exists a product $ A_{s_1}\cdots A_{s_l} $ of matrices from $ \mathcal{M} $ such that $ (A_{s_1}\cdots A_{s_l})_{ik}>0 $ and $ (A_{s_1}\cdots A_{s_l})_{jk}>0 $. Theorem 1 in \cite{alpin2013} states that the following condition is sufficient for an irreducible matrix set $\mathcal{M}$ to be primitive: for all indices $i,j$, there exists an index $k$ and a product $M$ of matrices from $\mathcal{M}$ such that $M_{ik}>0$ and $M_{jk}>0$. Therefore, we conclude.
\end{proof}

\section{The k-rendezvous time and a recurrence relation for its upper bound}\label{sec:krt}
In this section, we define the $ k $-rendezvous time of a primitive set of $ n\times n $ NZ matrices, we find an upper bound $ U_k(n) $ on it, and we prove a recurrence relation for $ U_k(n) $. 
\begin{definition}
Let $ \mathcal{M} $ be a primitive set of $ n\times n $ NZ matrices and $ k $ an integer such that $ 2\leq k\leq n $. We define the \emph{$k$-rendezvous time} ($k$-RT) to be the length of the shortest product of matrices from $ \mathcal{M} $ having a column or a row with $ k $ positive entries and we denote it by $ rt_k(\mathcal{M}) $. We indicate with $ rt_k(n) $ the maximal value of $ rt_k(\mathcal{M}) $ among all the primitive sets $\mathcal{M}$ of $ n\times n $ NZ matrices.
\end{definition}
 Our goal is to find, for any $ n\geq 2 $ and $ 2\leq k\leq n $, a function $ U_k(n) $ such that $  rt_k(n)\leq U_k(n) $.
\begin{definition}\label{defn:ank}
Let $ n$ and $k $ be two integers such that $ n\geq 2 $ and $ 2\leq k\leq n-1 $. We denote by $\mathcal{S}_n^k$ the set of all the $n \times n$ NZ matrices having every row and column of weight at most $k$ and at least one column of weight exactly $k$. 
For any $A \in \mathcal{S}_n^k$, let $\mathcal{C}_A$ be the set of the indices of the columns of $A$ having weight equal to $k$. We define $ a_k^n(A)=\min_{c\in \mathcal{C}_A}| \lbrace i: supp( A_{*i}) \nsubseteq  supp( A_{*c}) \rbrace|$ and $ a_k^n=\min_{A\in \mathcal{S}_n^k} a_k^n(A)$. 
\end{definition}
In other words, $ a_k^n(A)$ is the minimum over all the indices $ c\in \mathcal{C}_A $ of the number of columns of $ A $ whose support is not contained in the support of the $ c $-th column of $ A $.
Since the matrices are NZ, i.e.\ zero rows are not allowed, it holds that for any $ A\in \mathcal{S}_n^k $, $ 1\leq a_k^n\leq a_k^n(A) $. 
\begin{example}
If $ k=1 $, then $\mathcal{S}_n^1$ is the set of $ n\times n $ permutation matrices. In this case for any $A \in \mathcal{S}_n^1$, it holds that $ \mathcal{C}_A=[n]$ and $ a_1^n(A)=n-1=a_1^n $. Consider now the following matrices:
\begin{equation*}
A=\begin{pmatrix}
1&0&1&0\\
1&1&0&0\\
0&0&0&1\\
0&0&1&0
\end{pmatrix},
B=\begin{pmatrix}
1&1&1&0\\
1&1&0&0\\
0&0&0&1\\
0&0&1&0
\end{pmatrix},
C=\begin{pmatrix}
1&0&0&0\\
0&1&0&0\\
0&0&0&1\\
0&0&1&0
\end{pmatrix}.
\end{equation*}
It holds that $ A\in \mathcal{S}^2_4 $, while $ B,C\notin \mathcal{S}^2_4 $ because $ B $ has a row of weight $ 3 $ and $ C $ has no column of weight $ 2 $. Moreover:
\begin{itemize}
\item $ \mathcal{C}_A=\lbrace 1,3\rbrace $, 
\item $| \lbrace i: supp( A_{*i}) \nsubseteq  supp( A_{*1}) \rbrace|=| \lbrace 3,4\rbrace|=2$,
\item $| \lbrace i: supp( A_{*i}) \nsubseteq  supp( A_{*3}) \rbrace|=| \lbrace 1,2,4\rbrace|=3$,
\end{itemize}
and so $ a_2^4(A)=2 $. Therefore it holds that $ 1\leq a_2^4\leq 2 $. On the other hand, consider the following matrix:
\[
\hat{A}=\begin{pmatrix}
1&0&1&0\\
1&0&0&1\\
0&1&0&0\\
0&1&0&0
\end{pmatrix}
\]
We have that $ \hat{A}\in \mathcal{S}^2_4 $, $ \mathcal{C}_{\hat{A}}=\lbrace 1,2\rbrace $ and $| \lbrace i: supp( A_{*i}) \nsubseteq  supp( A_{*1}) \rbrace|=| \lbrace 2\rbrace|=1$, so $ a_2^4(\hat{A})=1 $. This implies that $ a_2^4=1$.

\end{example}

The following theorem shows that for every $ n\geq 2 $, we can recursively define a function $ U_k(n)\geq rt_k(n) $ on $ k$ by using the term $ a_k^n $.
\begin{theorem}\label{recurrence}
Let $ n\geq 2 $ be an integer. The following recursive function $ U_k(n) $ is such that for all $ 2\leq k < n $, it holds that $ rt_k(n)\leq U_k(n) $:
\begin{equation}\label{eq:simple_rec}
\begin{cases}
    U_2(n) = 1 \\
    U_{k+1}(n) = U_k(n) + n(1+n-a_k^{n})/2 & \text{for } 2 \leq k \leq n-1.
  \end{cases}
\end{equation}
\end{theorem}

\begin{proof}
We prove the theorem by induction.
Let $ k=2 $. Any primitive set of NZ matrices must have a matrix with a row or a column with two positive entries, as otherwise it would be made of just permutation matrices and hence it would not be primitive. This trivially implies that $rt_2(n)= 1\leq U_2(n)$.\\
Suppose now that $ rt_k(n)\leq U_k(n) $, we show that $ rt_{k+1}(n)\leq U_{k+1}(n) $. We remind that $ \mathcal{M}^d $ denotes the set of all the products of matrices from $ \mathcal{M} $ having length $ \leq d $. If in $\mathcal{M}^{rt_k(\mathcal{M}) + n-1}$ there exists a product having a column or a row with $ k+1 $ positive entries then $ rt_{k+1}(\mathcal{M})\leq rt_k(\mathcal{M}) + n-1 \leq U_{k+1}(n) $. Suppose now that this is not the case. This means that in $\mathcal{M}^{rt_k(\mathcal{M}) + n-1}$ every matrix has all the rows and columns of weight at most $ k $. Let $ A\in \mathcal{M}^{rt_k(\mathcal{M})} $ be a matrix having a row or a column of weight $ k $, and suppose it is a column. The case when $ A $ has a row of weight $ k $ will be studied later. 
By Lemma \ref{lem:prem_col} applied on the matrix $ A $, for every $i\in [n]$ there exists a matrix $ W_i\in \mathcal{M}^{rt_k(\mathcal{M}) + n-1} $ having the $ i $-th column of weight $ k $ (and all the other columns and rows of weight $ \leq k $). Every $ W_i $ has at least $ a_k^n $ columns (see Definition \ref{defn:ank}) whose support is not contained in the support of the $ i $-th column of $ W_i $: we pick $ a_k^n $ indices of these columns and we denote them by $c_i^1, c_i^2, \dots, c_i^{a_k^n}$. 
Notice that any product $ B $ of matrices from $ \mathcal{M} $ of length $ l $ such that $ B_{is}>0 $ and $ B_{c_i^js}>0 $ for some $ s\in [n] $ and $ j\in [a_k^n] $ would imply that $ W_iB $ has the $ s $-th column of weight at least $ k+1 $ and so $ rt_{k+1}(\mathcal{M})\leq rt_k(\mathcal{M}) + n-1+l $. We now want to minimize this length $ l $ over all $ i,s\in [n] $ and $ j\in [a_k^n] $: we will prove that there exists $ i,s\in [n] $ and $ j\in [a_k^n] $ such that $ l\leq n(n-1-a_k^n)/2 +1 $. 
To do this, we consider the pair digraph $ \mathcal{PD}(\mathcal{M})=(\mathcal{V}, \mathcal{E}) $ (see Definition \ref{defn:sg}) and the vertices
\begin{equation}\label{eq:list}
(1, c_1^1), (1, c_1^2),\dots ,(1, c_1^{a^n_k}),
		(2, c_2^1), \dots ,(2, c_2^{a^n_k}), \dots
		,(n, c_n^1), \dots ,(n, c_n^{a^n_k}).
\end{equation}
By Lemma \ref{lem:sg}, for each vertex in Eq.(\ref{eq:list}) there exists a path in $ \mathcal{PD}(\mathcal{M}) $ connecting it to a singleton. By the same lemma, a path of length $ l $ from $ (i,c_i^{j}) $ to a singleton $ (s,s) $ would result in a product $ B_j $ of matrices from $ \mathcal{M} $ of length $ l $ such that $ W_iB_j $ has the $ s $-th column of weight at least $ k+1 $. We hence want to estimate the minimal length among the paths connecting the vertices in Eq.(\ref{eq:list}) to a singleton. 
Notice that Eq.(\ref{eq:list}) contains at least $\left\lceil na_k^n/2 \right\rceil$ different elements, since each element occurs at most twice. 
It is clear that the shortest path from a vertex in the list (\ref{eq:list}) to a singleton does not contain any other element from that list.
The vertex set $ \mathcal{V} $ of $ \mathcal{PD}(\mathcal{M}) $ has cardinality $ n(n+1)/2 $ and it contains $ n $ vertices of type $ (s,s) $.  
 It follows that the length of the shortest path connecting some vertex from the list (\ref{eq:list}) to some singleton is at most of $n(n+1)/2-n - \left\lceil na_k^n/2 \right\rceil + 1\leq  n(n-1-a_k^n)/2+1 $. In view of what was said before, we have that there exists a product $ B $ of matrices from $ \mathcal{M} $ of length $ \leq n(n-1-a_k^n)/2+1 $ and $ i\in [n] $ s.t.\ 
 $ W_iB_j $ has a column of weight at least $ k+1 $. Since $ W_iB_j $ belongs to $\mathcal{M}^{rt_k(\mathcal{M}) + n-1+ n(n-1-a_k^n)/2+1}  $, it follows that $ rt_{k+1}(\mathcal{M})\leq rt_k(\mathcal{M})+ n(n+1-a_k^n)/2\leq U_{k+1}(n) $. \\
Suppose now $ A\!\in\! \mathcal{M}^{rt_k(\mathcal{M})} $ has a row of weight $ k $. We can use the same argument as above on the matrix set $ \mathcal{M}^{\top} $ made of the transpose of all the matrices in $ \mathcal{M} $. 
\end{proof}
Notice that the above argument stays true if we replace $a_k^n$ by a function $b(n,k)$ such that for all $ n\geq 2 $ and $ 2\leq k\leq n-1 $, it holds that $ 1\leq b(n,k)\leq a_k^n  $. It follows that Eq.(\ref{eq:simple_rec}) still holds true if we replace $ a_k^n $ by $ b(n,k) $. 


\subsection{Solving the recurrence}\label{sec:impro}
We now find an analytic expression for a lower bound on $a_k^n$ and we then solve the recurrence (\ref{eq:simple_rec}) in Theorem \ref{recurrence} by using this lower bound. We then show that this is the best estimate on $ a_k^n $ we can hope for.

\begin{lemma}\label{lem:lower_bound_akn}
Let $n $, $ k $ be two integers such that $ n\geq 2 $ and $2 \le k \le n-1$, and let $a_k^n$ as in Definition \ref{defn:ank}. It holds that $a_k^n \ge \max\lbrace n-k(k-1)-1,\lceil(n-k)/k\rceil, 1\rbrace$.
\end{lemma}
\begin{proof}
We have that $ a_k^n\geq 1 $ since $ k\leq n-1 $ and the matrices are NZ.
Let now $A \in \mathcal{S}^k_n$ (see Definition \ref{defn:ank}) and let $ a $ be one of its columns of weight $ k $. Let $ \zeta=supp( a) $; by assumption, the rows of $ A $ have at most $ k $ positive entries, so there can be at most $ (k-1)k $ columns of $ A $ different from $ a $ whose support is contained in $ \zeta $. Therefore, since $ A $ is NZ, there must exist at least $n - k(k-1) - 1 $ columns of $ A $ whose support is not contained in $ \zeta $  and so $a_k^n \ge  n-k(k-1)-1$.\\
Let again $A \in \mathcal{S}^k_n$ and let $ a $ be one of its columns of weight $ k $. Let $ \xi=[n]\setminus supp( a) $; $ \xi $ has cardinality $ n-k $ and since $ A $ is NZ, for every $ s\in \xi $ there exists $ s'\in [n] $ such that $ A_{ss'}>0 $. By assumption each column of $A$ has weight of at most $k$, so there must exist at least $ \lceil (n-k)/k\rceil $ columns of $ A $ different from $ a $ whose support is not contained in $ supp( a) $.
It follows that $a_k^n \ge \lceil(n-k)/k\rceil$.
\end{proof}
In view of the fact that the following inequalities hold:
\begin{enumerate}
\item $\lceil(n-k)/k\rceil \geq (n-k)/k $,
\item $n-k(k-1)-1 \geq (n-k)/k$ for $k \leq \lfloor\sqrt{n}\rfloor$,
\item $(n-k)/k \geq 1$ for $k \leq \lfloor n/2\rfloor$,
\end{enumerate}
the recursion (\ref{eq:simple_rec}) with $ a_k^n $ replaced by $\max\{n-k(k-1)-1,(n-k)/k, 1\}$ now reads as
\begin{equation}\label{second rec}
   B_{k+1}(n) =
  \begin{cases}
  1 &\text{if }k=1,\\
     B_k(n) + n(1+ k(k-1)/2)& \text{if }2\leq k \leq \lfloor \sqrt{n}  \rfloor, \\
     B_k(n) + n(1 + n(k-1)/2k)& \text{if }\lfloor \sqrt{n}  \rfloor+1  \leq k \leq \lfloor n/2  \rfloor, \\
     B_k(n) + n^2/2 & \text{if }\lfloor n/2  \rfloor+1 \leq k\leq n-1, \\
  \end{cases}\,\,  
\end{equation}
where we have denoted by $ B_k(n) $ the function solving (\ref{eq:simple_rec}) with $ a_k^n=\max\{n-k(k-1)-1,(n-k)/k, 1\}$.
The following proposition shows the solution of the recursion (\ref{second rec}).
\begin{proposition}\label{prop:rtk}
Equation (\ref{second rec}) is fulfilled by the following function:
\begin{equation}\label{eq:finalresult}
 B_{k}(n)\!=\!
  \begin{cases}
    \dfrac{n(k^3 - 3k^2 + 8k - 12)}{6} + 1 & \text{if }2\leq k \leq \lfloor \sqrt{n}  \rfloor,\\
     B_{\lfloor \sqrt{n}  \rfloor}(n) +\dfrac{n(n+2)(k-\lfloor\sqrt{n}\rfloor)}{2}-\dfrac{n^2}{2}\!\!\! \sum\limits_{i=\lfloor\sqrt{n}\rfloor}^{k-1} \!\!\dfrac{1}{i} & \text{if }  \lfloor \sqrt{n}  \rfloor+1 \leq k \leq \lfloor \frac{n}{2}  \rfloor ,\\
     B_{\lfloor \frac{n}{2}  \rfloor}(n) + \dfrac{(k-\lfloor \frac{n}{2} \rfloor)n^2}{2} & \text{if }\lfloor \frac{n}{2}  \rfloor+1\leq k\leq n\, .
  \end{cases}
\end{equation}
Therefore, for any constant $ k$ s.t. $k\leq \sqrt{n}  $, the k-rendezvous time $ rt_k(n) $ grows at most linearly in $ n $.
\end{proposition}

\begin{proof}
If $ 2\leq k \leq \lfloor \sqrt{n}  \rfloor $, let $C_k(n) = B_k(n)/n$. By Eq.(\ref{second rec}), it holds that $C_{k+1}(n)-C_k(n) = 1 + k(k-1)/2$. By setting $C_k(n) = \alpha k^3 + \beta k^2 + \gamma k + \delta$, it follows that $3\alpha k^2 + (3\alpha + 2\beta)k + \alpha + \beta + \gamma = k^2/2 - k/2 + 1$. Since this must be true for all $k$, by equating the coefficients we have that $C_k(n) = k^3/6 - k^2/2 + 4k/3 + \delta$. Imposing the initial condition $ B_2(n)=1$ gives finally the desired result $ B_k(n) = n(k^3 - 3k^2 + 8k - 12)/6 + 1$.\\
If $ \lfloor \sqrt{n}  \rfloor+1  \leq k \leq \lfloor n/2  \rfloor $, let again $C_k(n) =  B_k(n)/n$. By Eq.(\ref{second rec}), it holds that $C_{k+1}(n)-C_k(n) = 1 + n(k-1)/2k$ and so $C_{k}(n) = C_{\lfloor\sqrt{n}\rfloor}(n) + (k-2)(1+n/2) -(n/2)\sum_{i=\lfloor\sqrt{n}\rfloor}^{k-1}i^{-1}$.
Since $  C_{\lfloor\sqrt{n}\rfloor}(n)= B_{\lfloor\sqrt{n}\rfloor}(n)/n $, it follows that $  B_{k}(n) =  B_{\lfloor\sqrt{n}\rfloor}(n) + (k-\lfloor\sqrt{n}\rfloor)n(n+2)/2 -(n^2/2)\sum_{i=\lfloor \sqrt{n}\rfloor}^{k-1}i^{-1} $.\\
If $ \lfloor n/2  \rfloor+1\leq k\leq n-1 $, by Eq.(\ref{second rec}) it is easy to see that $  B_k(n)= B_{\lfloor n/2  \rfloor}(n)+(k-\lfloor n/2  \rfloor)n^2/2 $, which concludes the proof.
\end{proof}
We now show that $ a_k^n\!=\!\max\{n-k(k-1)-1,\lceil(n-k)/k\rceil, 1\}$, and so we cannot improve the upper bound on $ rt_k(n) $ by improving our estimate of $ a_k^n $.
\begin{lemma}
Let $n$ and $ k $ be two integers such that $ n\geq 2 $ and $2 \le k \le n-1$. It holds that: 
 \begin{equation*}
1\leq a_k^n \le u(n,k):=
\begin{cases}
   n - k(k-1) - 1  & \text{if } n - k(k-1)-1\geq \lceil (n-k)/k\rceil, \\
    \lceil (n-k)/k\rceil  & \text{otherwise. }
\end{cases}\,\, 
 \end{equation*}
 \end{lemma}
 \begin{proof}
We need to show that for every $n \geq 2$ and $2 \le k \le n-1$, there exists a matrix $A \in \mathcal{S}_n^k$ such that $a_k^n(A) = u(n,k)$ (see Definition \ref{defn:ank}).
We define the matrix $ C_i^{m_1\times m_2} $ as the $ m_1\times m_2$ matrix having all the entries of the $ i $-th column equal to $ 1 $ and all the other entries equal to $ 0 $, and the matrix $ R_i^{m_1\times m_2} $ as the $ m_1\times m_2$ matrix having all the entries of the $ i $-th row equal to $ 1 $ and all the other entries equal to $ 0 $. We indicate by $ \mathbf{0}^{m_1\times m_2} $ the $ m_1\times m_2$ matrix having all its entries equal to zero and by $ \mathbf{I}^{m\times m} $ the $ m\times m$ identity matrix.
Let $v_k^n = \lceil(n-k)/k\rceil+1$ and $ q=n \mod k $.\\
Suppose that $ n - k(k-1)-1\geq \lceil (n-k)/k\rceil $ and set $ \alpha=n - k(k-1)-1- \lceil (n-k)/k\rceil $.
Then the following matrix $ \hat{A} $ is such that $a_k^n(\hat{A}) = n - k(k-1)-1=u(n,k)$:
\begin{equation*}
\hat{A}=\left[
\begin{array}{c|cccc|c}
C_1^{k\times v_k^n} & R_1^{k\times (k-1)} & R_2^{k\times (k-1)} &\!\!\! \cdots \!\!\!& R_{k}^{k\times (k-1)} &  \\
C_2^{k\times v_k^n} &  & &  & &    \\
\vdots &  &  \mathbf{0}&\!\!\!\!\!\!\!\!\!\!\!\!^{(n-k)\times [k(k-1)]}& &  D   \\
C_{v_k^n-1}^{k\times v_k^n} &  & &  & &     \\
C_{v_k^n}^{q\times v_k^n} &  & &  & &     \\
\end{array}
\right], D=\left[
\begin{array}{c}
 \mathbf{0}^{k\times\alpha} \\
 \mathbf{I}^{\alpha\times\alpha}   \\
 \mathbf{0}^{(n-k-\alpha)\times \alpha}     \\
\end{array}
\right].
\end{equation*}
Indeed by construction, the first column of $ \hat{A} $ has exactly $ k $ positive entries. The columns of $ \hat{A} $ whose support is not contained in $ \hat{A}_{*1} $ are the columns $ \hat{A}_{*i} $ for $ i=2,\dots, v_k^n $ and all the columns of $ D $. In total we have $ \lceil(n-k)/k\rceil + \alpha= n - k(k-1)-1 $ columns, so it holds that $ a_k^n(\hat{A})=n - k(k-1)-1 $.

Suppose that $ n - k(k-1)-1\leq \lceil (n-k)/k\rceil $. Then the following matrix $ \tilde{A} $ is such that $a_k^n(\tilde{A}) = \lceil(n-k)/k\rceil=u(n,k)$:
\begin{equation*}
\tilde{A}=\left[
\begin{array}{c|ccccc}
C_1^{k\times v_k^n} & R_1^{k\times (k-1)} & R_2^{k\times (k-1)} & \cdots & R_{k-1}^{k\times (k-1)} & R_{k}^{k\times (n-v_k^n-(k-1)^2)} \\
\hline
C_2^{k\times v_k^n} &  & &  & &     \\
\vdots &  & &  & &     \\
C_{v_k^n-1}^{k\times v_k^n} &  & &  \mathbf{0}^{(n-k)\times(n-v_k^n)} & &     \\
C_{v_k^n}^{q\times v_k^n} &  & &  & &     \\
\end{array}
\right].
\end{equation*}
Indeed by construction, the first column of $ \tilde{A} $ has exactly $ k $ positive entries and the columns of $ \tilde{A} $ whose support is not contained in $ \tilde{A}_{*1} $ are the columns $ \tilde{A}_{*i} $ for $ i=2,\dots, v_k^n $. Therefore it holds that $ a_k^n(\tilde{A})=v_k^n-1=\lceil(n-k)/k\rceil $.
 \end{proof}

\subsection{Numerical results for the upper bound $ B_k(n) $}\label{sec:num_B}
We now present some numerical results that compare the theoretical bound $ B_k(n)$ on $ rt_k(n) $ of Eq.(\ref{eq:finalresult}) with either the exact $k$-RT or an heuristic approximation of it when the computation of the exact value is not computationally feasible for some primitive sets. In Figure \ref{fig:test} we compare our bound with the real $ k $-RT of the primitive sets $\mathcal{M}_{CPR}$ and $\mathcal{M}_K$ reported here below:
\begin{align*}
&\mathcal{M}_{CPR} = 
\left\{
\begin{pmatrix} 0 & 0 & 1 & 0\\ 1 & 1 & 0 & 0\\ 1 & 0 & 0 & 0\\ 0 & 0 & 0 & 1 \end{pmatrix},
\begin{pmatrix} 1 & 0 & 0 & 0\\ 0 & 0 & 1 & 0\\ 0 & 0 & 0 & 1\\ 0 & 1 & 0 & 0 \end{pmatrix}
\right\},\\
&\mathcal{M}_{K} = 
\left\{
\begin{pmatrix} 1 & 0 & 0 & 1 & 0 & 0\\ 0 & 1 & 0 & 0 & 0 & 0\\ 0 & 0 & 1 & 0 & 0 & 0\\ 0 & 0 & 0 & 0 & 1 & 0\\ 0 & 0 & 0 & 1 & 0 & 0\\ 0 & 0 & 0 & 0 & 0 & 1 \end{pmatrix},
\begin{pmatrix} 0 & 0 & 0 & 0 & 1 & 0\\ 0 & 0 & 1 & 0 & 0 & 0\\ 0 & 0 & 0 & 1 & 0 & 0\\ 0 & 1 & 0 & 0 & 0 & 0\\ 0 & 0 & 0 & 0 & 0 & 1\\ 1 & 0 & 0 & 0 & 0 & 0 \end{pmatrix}
\right\}.
\end{align*}
The sets $\mathcal{M}_K$ and $\mathcal{M}_{CPR}$ are primitive sets of matrices that are based, respectively, on the Kari automaton \cite{Kari2001} and the \v{C}ern\'{y}-Pirick\'{a}-Rozenaurov\'{a} automaton \cite{CernyPiricka}, which are well known synchronizing automata with large (quadratic) reset threshold. The sets \ $\mathcal{M}_K$ and $\mathcal{M}_{CPR}$ have been created by adding a $ 1 $ in one single column of the matrix corresponding to one particular letter of the Kari and the \v{C}ern\'{y}-Pirick\'{a}-Rozenaurov\'{a} automaton respectively, in order to make the matrix set primitive.
 We can see that for small values of $k$, the upper bound is fairly close to the actual value of $rt_k(\mathcal{M})$.
\begin{figure}[h]
\centering
\subfigure[][$n=4$, $\mathcal{M} = \mathcal{M}_{CPR}$]{\includegraphics[scale=0.15]{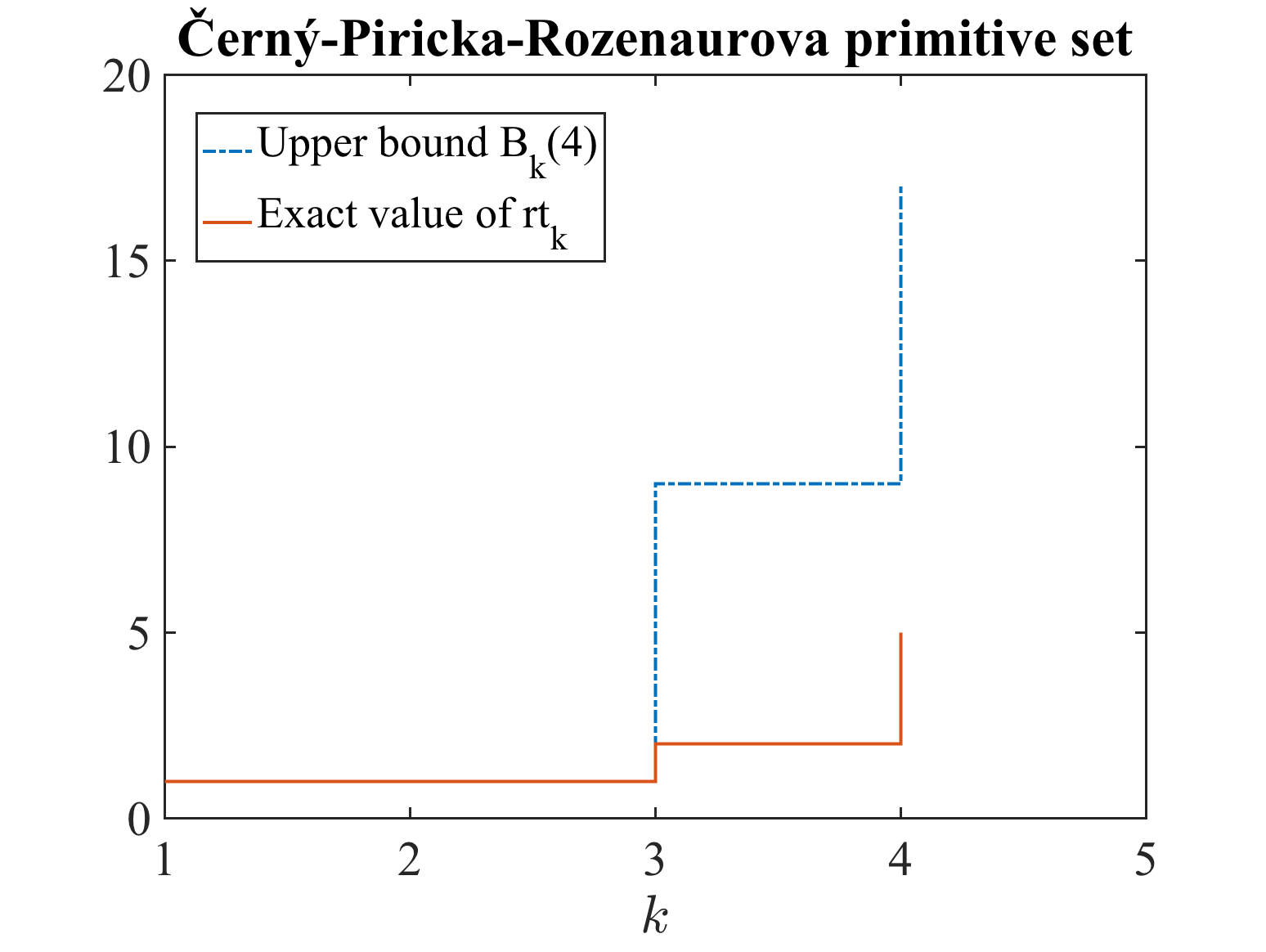}}\subfigure[][$n=6$, $\mathcal{M} = \mathcal{M}_{K}$]{\includegraphics[scale=0.15]{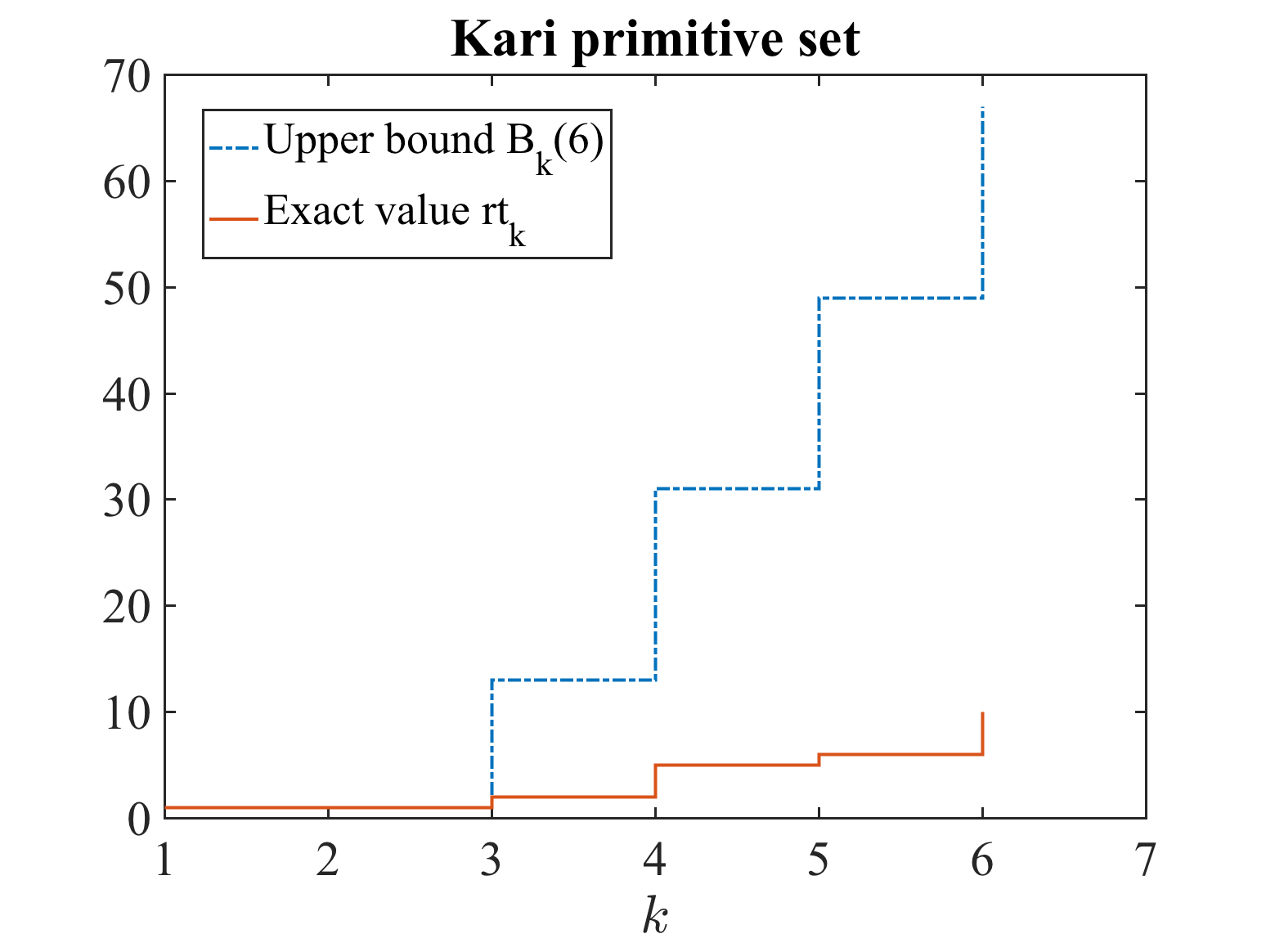}}
\caption{Comparison between the bound $ B_k(n)$, valid for all primitive NZ sets, and $rt_k(\mathcal{M})$ for (a) $\mathcal{M} = \mathcal{M}_{CPR}$ and (b) $ \mathcal{M} = \mathcal{M}_{K}$.}
\label{fig:test}
\end{figure}
When $ n $ is large, computing the $k$-RT for every $ 2\leq k\leq n $ becomes hard, so we compare our upper bound on the $ k $-RT with a method for approximating it. The \emph{Eppstein heuristic} is a greedy algorithm developed by Eppstein in \cite{Epp} for approximating the reset threshold of a synchronizing automaton. Given a primitive set $ \mathcal{M} $ of binary NZ matrices, we can apply a slightly modified Eppstein heuristic to obtain, for any $k$, an upper bound on $ rt_k(\mathcal{M}) $. 
\begin{algorithm}[h!]\label{alg}
\caption{Pseudo-code for the modified Eppstein heuristic}
\SetAlgoNoLine
\BlankLine
\BlankLine
\KwIn{A primitive matrix set $\mathcal{M}$}
\KwOut{A matrix $A$ of elements from $\mathcal{M}$ with a positive column}
  $A \longleftarrow$ arg$\max_{X \in \mathcal{M}}\max_{i\in [n]} |\textrm{supp}(X_{\ast i})|$;
  
  $i \longleftarrow $ arg$\max_{j\in [n]} |\textrm{supp}(A_{\ast j})|$
  
  $S \longleftarrow \textrm{supp}(A_{\ast i})$
 
 \While{$S \neq [n]$}{$\mathcal{C} \longleftarrow \{j\in [n] : \textrm{supp}(A_{\ast j}) \nsubseteq S\}$;
 
    $j^{\ast} \longleftarrow $arg$\min_{j\in \mathcal{C}} d_{\mathcal{PD}(\mathcal{M})}[(i,j),(i,i)]$ \tcc*{See Remark 3.7.}
    
    $A_{p_1},\dots,A_{p_l} \longleftarrow$ labels of shortest path from $(i,j)$ to $(i,i)$;
    
    $A \longleftarrow A A_{p_1} \cdots A_{p_l}$;
    
    $S \longleftarrow \textrm{supp}(A_{\ast i})$}
  \KwRet{A} ;
\end{algorithm}

This modified Eppstein heuristic is formalized in Algorithm 1, where for any nodes $(i,j)$ and $(k,l)$ in $\mathcal{PD(M)}$ (see Definition \ref{defn:sg}), we denote by $d_{\mathcal{PD(M)}}[(i,j),(k,l)]$ the length of the shortest path from $(i,j)$ to $(k,l)$ in $\mathcal{PD(M)}$. The algorithm looks for the matrix $ A $ in the set having the column with the maximal number of positive entries; then by making use of the pair graph and Lemma \ref{lem:sg}, it looks for the shortest product $ B $ of matrices in the set such that the number of positive entries of $ (AB)_{*i} $ is strictly greater than the one in $ A_{*i} $. It iterates this procedure until obtaining a matrix with a positive column.
Since the algorithm increases the weight of the column $i$ at each iteration of the while loop, it provides an upper bound on the $k$-RT for all $2\leq k\leq n$ by producing a (reasonably short) product whose $i$-th column is of weight $\geq k$.
At the same time, we also check the weights of the rows of $A$ in case the algorithm happens to produce a larger row weight than the maximal column weight, thus improving the bound on the $k$-RT.

\begin{remark}
In our implementation we look for the shortest path to a specific singleton $(i,i)$, whereas it would in general be better to find the shortest path to any singleton. However, one can show that in the case where among the columns of the matrices of the set $\mathcal{M}$ there is only one column with two positive entries, the two implementations are equivalent. Since this is the case for all the matrix sets used in the numerical experiments of this paper, the choice of implementation was based on considerations of simplicity.
\end{remark}

We now consider the primitive sets with quadratic exponent presented by Catalano and Jungers in \cite{CatalanoJALC}, Section 4; here we denote these sets by  $\mathcal{M}_{C_n}$, where $n$ is the matrix dimension. Figure \ref{fig:test2} compares our theoretical upper bound $ B_k(n) $ with the results of the Eppstein heuristic on the $k$-RT of $\mathcal{M}_{C_n}$ for $ n=10,15,20,25 $ and $ 2\leq k\leq n $. 
We can see again that for small values of $k$ our generic upper bound is fairly close to the Eppstein heuristic of $rt_k(\mathcal{M}_{C_n})$.

Finally, Figure \ref{fig:test3} compares the theoretical upper bound $ B_k(n) $ with the results of the Eppstein heuristic on the $k$-RT of the family $\mathcal{M}_{C_n}$ for fixed $k=4$ and $21\leq n\leq 30$. It can be noticed that $ B_k(n)$ does not increase very rapidly as compared to the Eppstein approximation.

\begin{figure}[h]
\centering 
\subfigure[][$n=10$, $\mathcal{M} = \mathcal{M}_{C_{10}}$]{\includegraphics[scale=0.18]{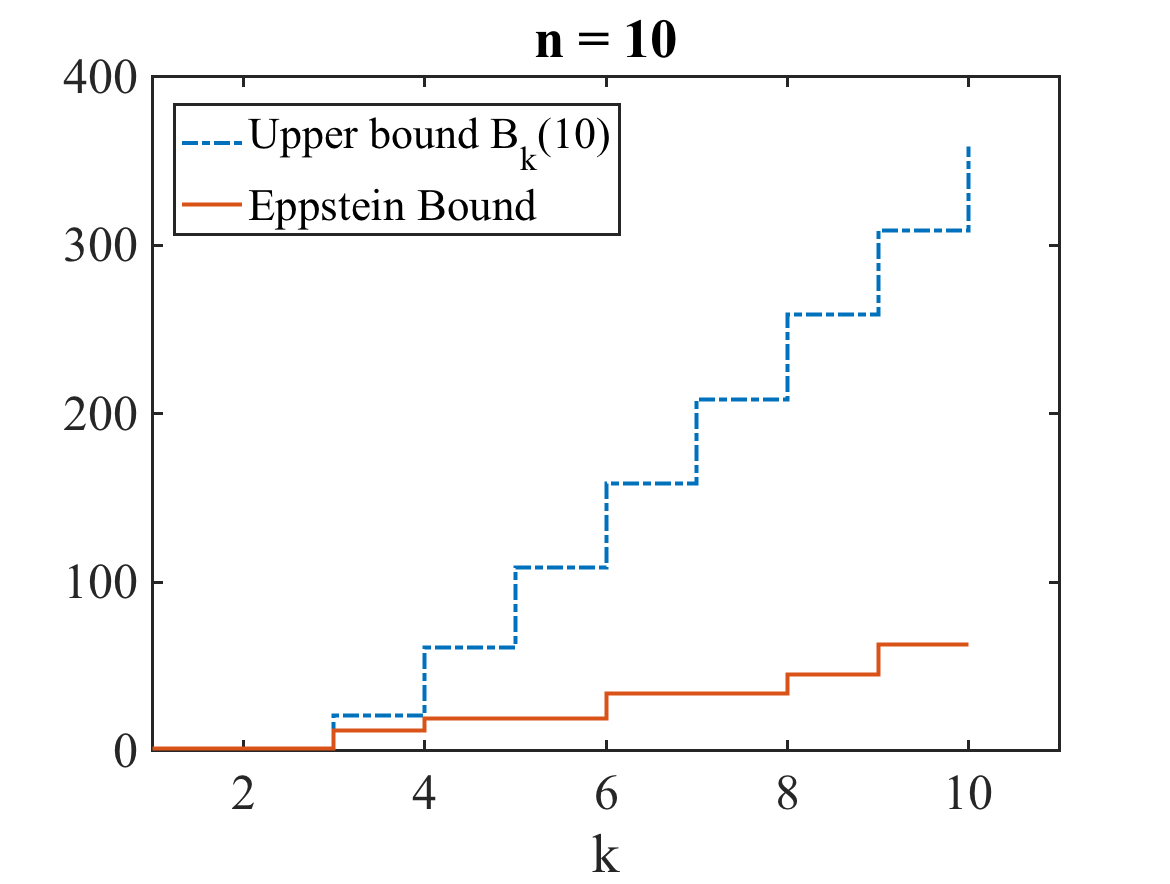}}
\subfigure[][$n=15$, $\mathcal{M} = \mathcal{M}_{C_{15}}$]{\includegraphics[scale=0.18]{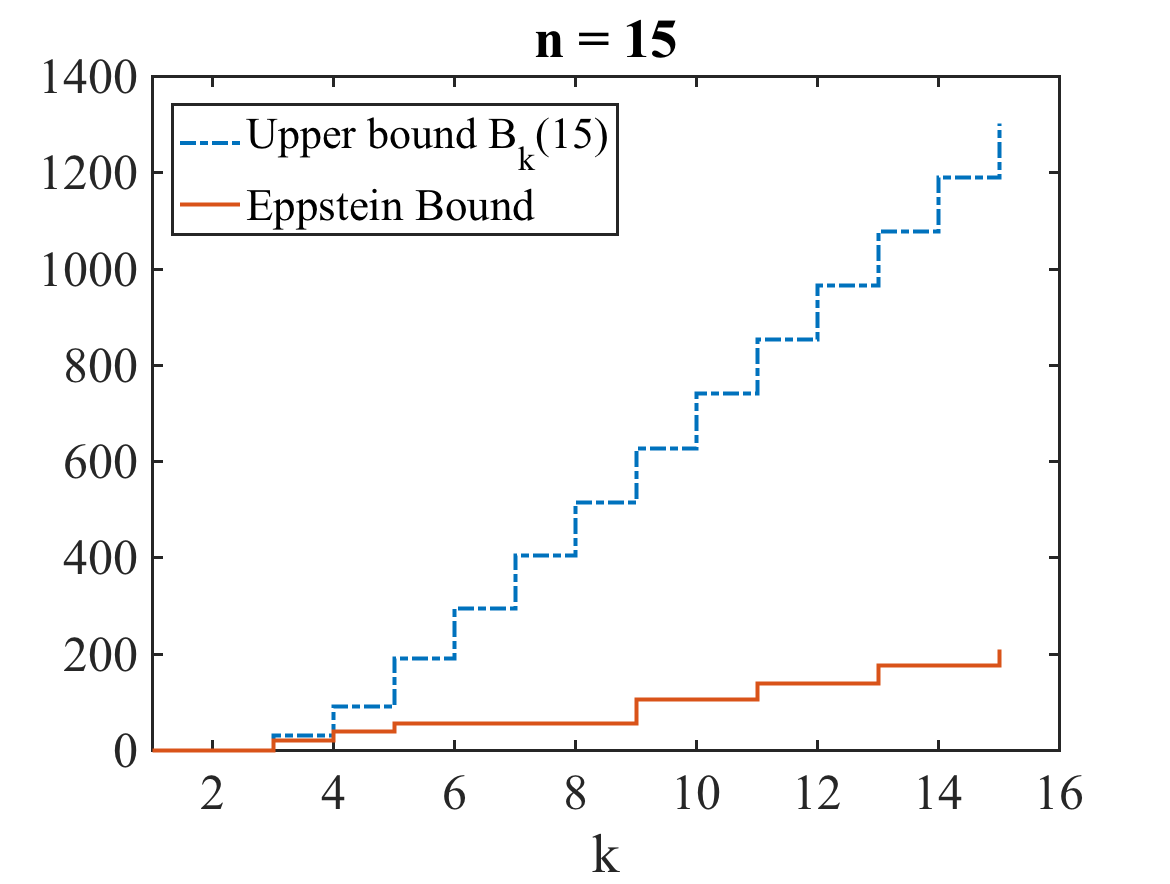}}\\
\subfigure[][$n=20$, $\mathcal{M} = \mathcal{M}_{C_{20}}$]{\includegraphics[scale=0.18]{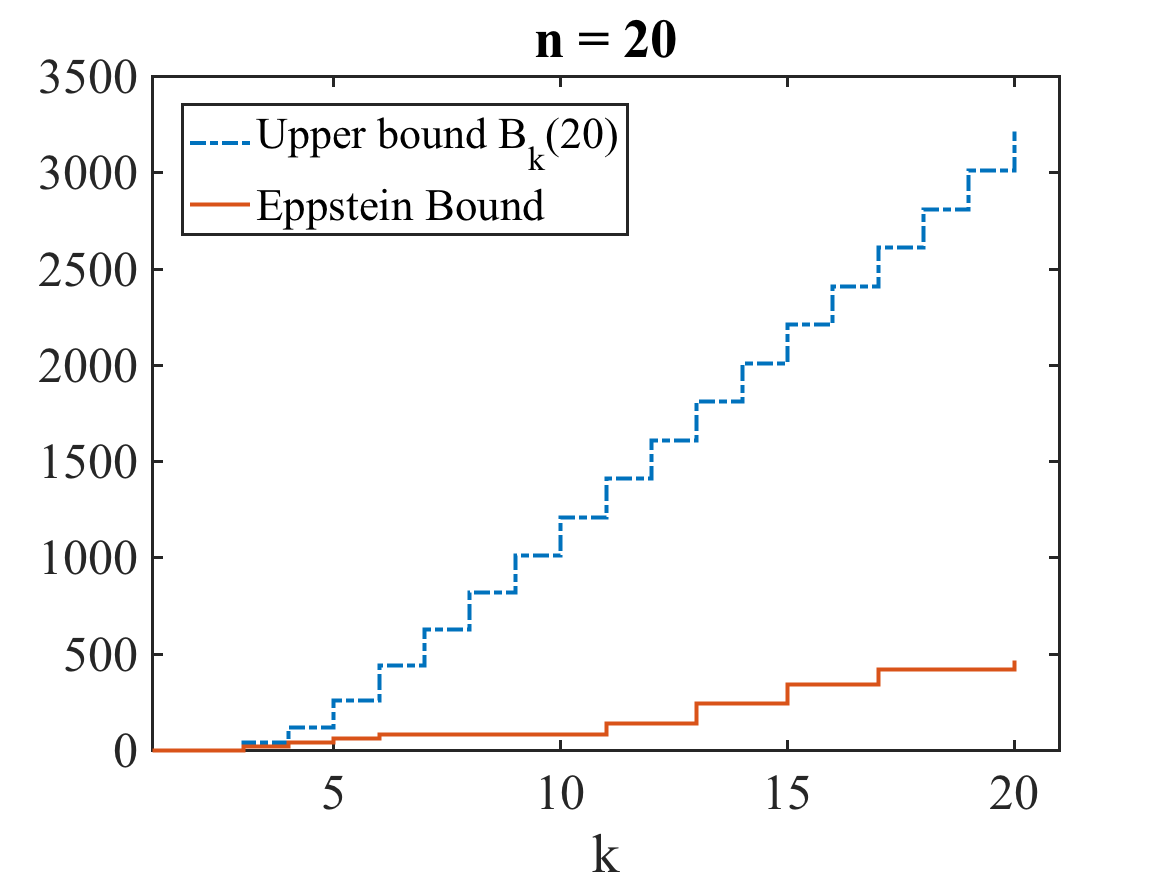}}
\subfigure[][$n=25$, $\mathcal{M} = \mathcal{M}_{C_{25}}$]{\includegraphics[scale=0.18]{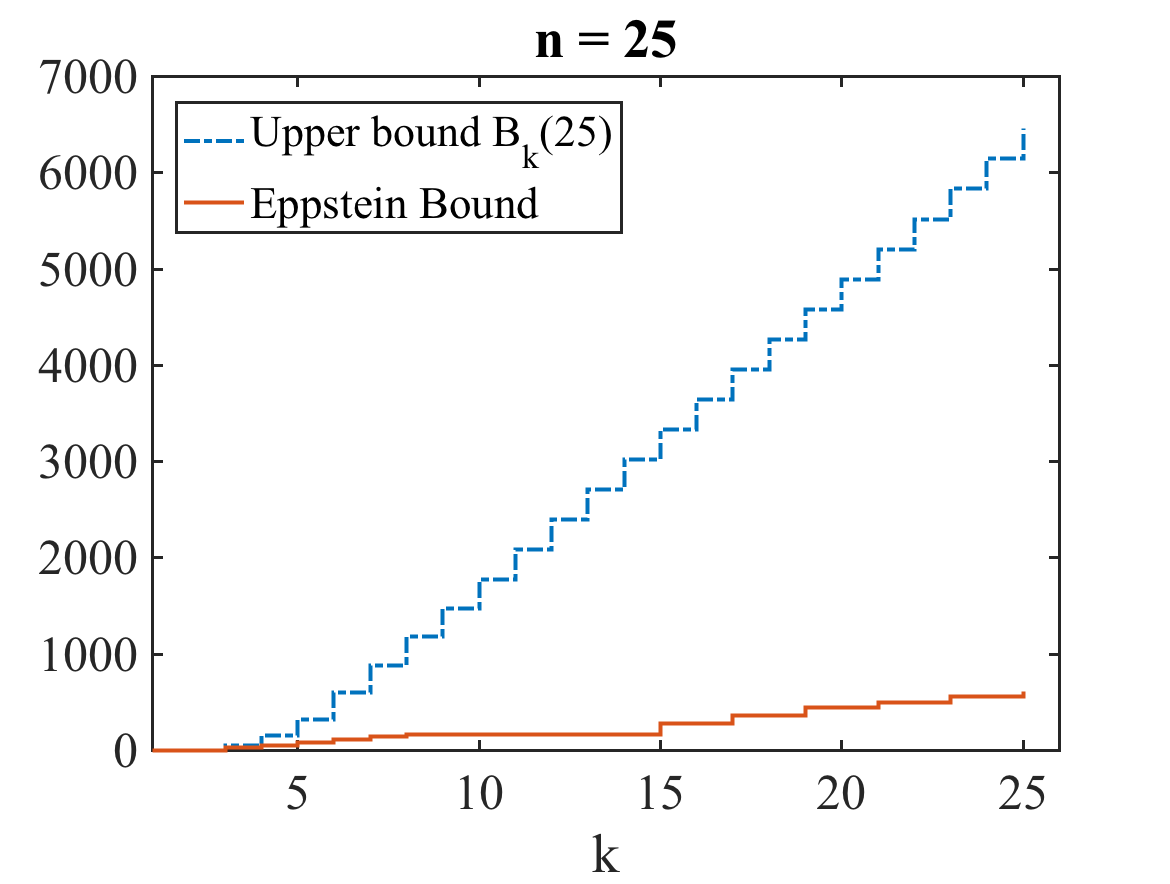}}  
\caption{Comparison between $B_k(n)$ and the Eppstein approx. of $rt_k(\mathcal{M})$, for (a) $\mathcal{M} = \mathcal{M}_{C_{10}}$, (b) $\mathcal{M} = \mathcal{M}_{C_{15}}$, (c) $\mathcal{M} = \mathcal{M}_{C_{20}}$, (d) $\mathcal{M} = \mathcal{M}_{C_{25}}$. We recall that $B_k(n)$ is a generic bound valid for all primitive NZ sets, while the Eppstein bound is computed on each particular set.}
\label{fig:test2}
\end{figure}
\begin{figure}[h!]
  \centering
  \includegraphics[scale=0.15]{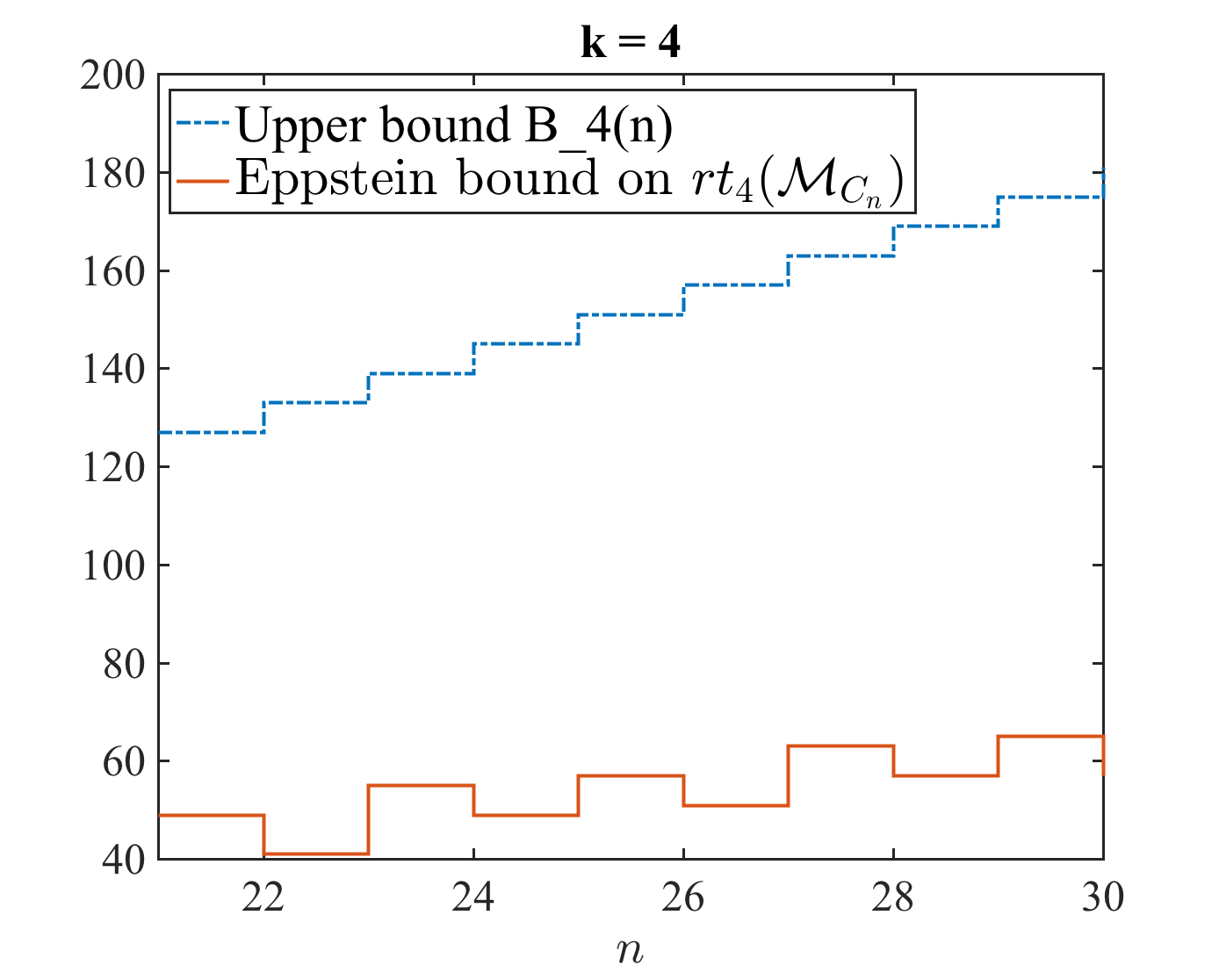}
\caption{Comparison between $B_k(n)$ and the Eppstein approx. of $rt_k(\mathcal{M}_{C_n})$ for $k=4$. We recall that $ B_k(n)$ is a generic bound valid for all primitive NZ sets, while the Eppstein bound is computed on each particular set.}
\label{fig:test3}
\end{figure}

\section{Improving the upper bound on $ rt_k(n) $}\label{sec:improv_ub}
In this section we propose a method for improving the upper bound $ B_k(n) $ in Eq.\ (\ref{eq:finalresult}) on $ rt_k(n) $. The idea behind this improvement is the following: 
Eq.\ (\ref{eq:finalresult}) comes from the recursive formulation (\ref{eq:simple_rec}), where we start from a matrix $ A $ with a column (or a row) of weight $ k $ and we look for the shortest product $ D $ that makes $ AD $ having a column (or a row) of weight at least $ k+1 $. This is done by `\textit{merging}' two columns of $ A $ by post-multiplying with matrix $ D $; namely, if $ c $ is the index of a column of $ A $ of weight $ k $ and $ i$ is the index of a column of $ A $ such that $supp(A_{*i})\nsubseteq supp(A_{*c}) $, we look for a (short) product $ D $ such that $ D_{ij}>0 $ and $ D_{cj}>0 $ for some $ j\in [n] $. In this way, since $ A $ and $ D $ are NZ, the $ j $-th column of $ AD $ has at least $ k+1 $ positive entries. What is not exploited in this approach is that if $ \vert supp( A_{*i})\setminus supp( A_{*c})\vert = l\geq 1$, then the $ j $-th column of $ AD $ has $ k+l $ positive entries, thus potentially providing an upper bound on the $ (k+l) $-RT as well as on $ (k+1) $-RT. 

In this section we take advantage of this fact to improve the upper bound $ B_k(n) $. Unfortunately, this refined upper bound comes as a solution of a much more complex recursive equation that is numerically solvable but that does not seem to be easily expressed in closed form. Therefore, although it is possible to prove that this new upper bound is smaller than or equal to $ B_k(n) $ for all values of $ n $ and $ 2\leq k\leq n $, we do not provide an analytic formula for it. Extensive numerical experiments show the behaviour of this new upper bound on $ rt_k(n) $ are presented in Section \ref{sec:num_F}.

We start with a technical definition, similar to Definition \ref{defn:ank}, that will be useful for finding the new upper bound. 
\begin{definition}\label{defn:Spkn}
Let $ n $, $ k $ and $ p $ be three integers s.t.\ $ n\geq 3 $, $ 2\leq k<n $, and $ 1\leq p\leq \min\lbrace k,n-k\rbrace $. We denote by $ \mathcal{S}^{k}_{n,p} $ the subset of matrices $ A $ in $ \mathcal{S}^{k}_n $ (see Definition \ref{defn:ank}) such that for any column index $ c\in\mathcal{C}_A $ and $ i\neq c $ it holds that $ \vert \lbrace supp(A_{*i})\setminus supp(A_{*c})\rbrace\vert \leq p$, and there exists $ c'\in\mathcal{C}_A $ and $ j\neq c' $ such that $ \vert \lbrace supp(A_{*j})\setminus supp(A_{*c'})\rbrace\vert = p$. For any matrix $ A\in  \mathcal{S}^{k}_{n,p}$, we denote by $ \mathcal{C}^p_A $ the subset of columns $ c $ in $ \mathcal{C}_A $ for which there exists $ i\neq c $ such that  $ \vert supp(A_{*i})\setminus supp(A_{*c})\vert = p$. \\ Finally, for any $ A\in  \mathcal{S}^{k}_{n,p}$ we set 
\begin{align*}
a^n_{{k},p}(A)=\min_{c\in \mathcal{C}^p_A} \vert\lbrace i: supp(A_{*i})\nsubseteq  supp(A_{*c})  \rbrace \vert\, ; \qquad a^n_{{k},p}=\min_{A\in \mathcal{S}^{k}_{n,p} }a^n_{{k},p}(A).
\end{align*}
\end{definition}
In words, $ a^n_{{k},p} $ describes the minimal number of columns that can be summed up to a column of weight $ k $ in any matrix in $ \mathcal{S}^{k}_{n,p} $ in order to increase its weight.

\begin{remark}\label{rem:S_npk}
Notice that $\mathcal{S}^{k}_n=\bigcup_{p=1}^{\min\lbrace k,n-k\rbrace} \mathcal{S}^{k}_{n,p}$ and that for any $ p $ such that $ 1\leq p\leq \min\lbrace k,n-k\rbrace $, it holds that $ a^n_{k,p} \geq a^n_k  $. By mimicking the reasoning used in the proof of Lemma \ref{lem:lower_bound_akn}, it is also easy to see that $ a^n_{k,p} \geq \lceil(n-k)/p\rceil$. By the same lemma, this implies that 
\begin{equation}\label{eq:ank_tilde}
a^n_{k,p} \geq \hat{a}^n_{k,p}:= \max \lbrace n-k(k-1)-1, \lceil (n-k)/p\rceil, 1\rbrace.
\end{equation}
\end{remark}
Let $ n\geq 3 $ and $ h $, $ k $ be two integers such that $ 2\leq h,k\leq n $.
 We set:
\begin{equation}\label{eq:=Okn}
O^k_{h}(n)=\max_{A\in \mathcal{S}^{h}_n}\min\lbrace d\in\mathbb{N}:\exists D\in\mathcal{M}^d \text{ s.t. } AD \text{ has a column of weight} \geq k\rbrace.
\end{equation}
Notice that $ O^k_{h}(n)=0 $ any time $ h\geq k $. In view of Eq.\ (\ref{eq:=Okn}), for any $ h $ and $ k $ s.t.\ $ 2\leq h\leq k\leq n$ it holds that \[ rt_k(n)\leq U_{h}(n)+ O^k_{h}(n)\leq B_{h}(n)+ O^k_{h}(n), \] where $ U_{h}(n) $ is defined in Eq.\ (\ref{eq:simple_rec}) and $ B_{h}(n) $ is defined in Eq.\ (\ref{second rec}). Since $B_{h}(n)+ O^k_{h}(n)= B_{k}(n) $  if $ k=h $, it follows that

\begin{equation}\label{eq:upp_ohk_min}
rt_k(n)\leq \min_{2\leq h\leq k}\lbrace B_{h}(n)+ O^k_{h}(n)\rbrace\leq B_{k}(n).
\end{equation}
Therefore the function $ \min_{2\leq h\leq k}\lbrace B_{h}(n)+ O^k_{h}(n)\rbrace$ improves the upper bound $B_{k}(n)  $. The rest of the section is devoted to finding a way to approximate $ O^k_{h}(n) $ for any $ n $, $ h $ and $ k $. 

Let $ p $ be an integer such that $ 1\leq p\leq \min\lbrace k,n-k\rbrace $. We set

\begin{equation*}
O^k_{h,p}(n)=\max_{A\in \mathcal{S}^{h}_{n,p}}\min \lbrace d\in\mathbb{N}: \exists D\in\mathcal{M}^d \text{ s.t. } AD \text{ has a column of weight} \geq k\rbrace\, . 
\end{equation*}
Notice again that $ O^k_{h,p}(n)=0 $ any time $ h\geq k $. It follows from Remark \ref{rem:S_npk} that 
\begin{equation}\label{eq:bound_Okn}
O^k_{h}(n)=\max_{ 1\leq p\leq \min\lbrace h,n-h\rbrace} O^k_{h,p}(n).
\end{equation}
The following result uses this fact to obtain a recursive upper bound for $ O^k_{h}(n) $ when $ h<k $ (for $h\geq k$ we already know that $ O^k_{h}(n)=0 $).


\begin{proposition}\label{prop:O_kn}
Let $ n $, $ k $ and $ h $ integers such that $ 2\leq h< k\leq n $. It holds that
\begin{equation}
O^k_{h}(n)\leq \max_{1\leq p\leq \min\lbrace h,n-h\rbrace}\min\left\lbrace  O^k_{h+p}(n)+\frac{n}{2}(n-1);O^k_{h+1}(n)+\frac{n}{2}(n+1-\hat{a}^n_{h,p})\right\rbrace.
\end{equation}
\end{proposition}

\begin{proof}
By Eq.\ (\ref{eq:bound_Okn}), it suffices to prove that $ O^k_{h,p}(n)\leq   O^k_{h+p}(n)+n(n-1)/2$ and $O^k_{h,p}(n)\leq O^k_{h+1}(n)+n(n+1-\hat{a}^n_{h,p})/2$. To prove the first inequality, we need to show that for any matrix $ A\in \mathcal{S}^{h}_{n,p} $ there exists a product $ D $ of length at most $ n(n-1)/2 $ such that $ AD $ has a column of weight at least $ h+p$. Let $ A\in \mathcal{S}^{h}_{n,p} $ and $ c\in \mathcal{C}^p_A $; $ A_{*c} $ has  $ h  $ positive entries. By the definition of $ \mathcal{C}^p_A  $ (see Definition \ref{defn:Spkn}), there exists an index $ i $ such that $ \vert supp(A_{*i})\setminus supp(A_{*c}) \vert=p $. In the pair digraph $ \mathcal{PD}(\mathcal{M}) $ (see Definition \ref{defn:sg}) there exists a path from the vertex $ (i,c) $ to a singleton of length at most $ (n(n+1)/2)-n=n(n-1)/2 $, in view of Lemma \ref{lem:sg} and the number of vertices in the pair digraph. By the same lemma, this means that there exists a product $ D $ of length at most $ n(n-1)/2 $ such that $ AD $ has a column of weight $\geq h+p$. By the definition of $ O^k_{h+p}(n) $ in Eq.\ (\ref{eq:=Okn}) it follows that $ O^k_{h,p}(n)\leq    O^k_{h+p}(n)+n(n-1)/2$. \\
To prove the second inequality, we need to show that for any matrix $ A\in \mathcal{S}^{h}_{n,p} $ there exists a product $ D $ of length at most $ n(n+1-\hat{a}^n_{h,p})/2 $ such that $ AD $ has a column of weight at least $ h+1$. This can be done in view of Eq.\ (\ref{eq:ank_tilde}) and by mimicking the proof of Theorem \ref{recurrence}. $ $
\end{proof}
The following theorem provides an upper bound $ \tilde U^k_{h}(n) $ on $ O^k_{h}(n) $ defined by a recurrence relation.

\begin{theorem}\label{thm:U_hk}
Let $ n $ and $ k $ be two integers such that $ 2\leq k\leq n $. We define the function $\tilde U^k_{h}(n)  $ for $  h\geq 2  $ in the following way:\\
if $ h\geq k $,
\[
\tilde U^k_{h}(n)=0\, ;
\]
if $ 2\leq h <k $,
\[
\tilde U^k_{h}(n)= \max\limits_{1\leq p\leq \min\lbrace h,n-h\rbrace}\min\left\lbrace   \tilde U^k_{h+p}(n)+\frac{n}{2}(n-1); \tilde U^k_{h+1}(n)+\frac{n}{2}(n+1-\hat{a}^n_{h,p})\right\rbrace\, .
\]
Then it holds that $ \tilde U^k_h(n)\geq O^k_h(n) $ for all $ h\geq 2 $.
In particular it holds that for every $ 2\leq k\leq n $:
\begin{equation}\label{eq:upp_rt_h}
rt_k(n)\leq F_{k}(n):= \min_{2\leq h\leq k}\lbrace B_{h}(n)+ \tilde U^k_{h}(n)\rbrace \leq B_{k}(n) .
\end{equation}
\end{theorem}

\begin{proof}
The theorem easily follows from Proposition \ref{prop:O_kn} and Eq. (\ref{eq:upp_ohk_min}).
\end{proof}

We conclude this section by noting that despite the fact that the function $ \tilde U^k_h(n) $ defined in Theorem \ref{thm:U_hk} seems difficult to obtain in closed form, it can be easily implemented and computed by dynamic programming. The numerical behavior of the function $ F_k(n) $ in Eq.\ (\ref{eq:upp_rt_h}), which represents the improved upper bound on $ rt_k(n) $, is the subject of the next section. These numerical results suggest the two following conjectures, for which we do not have a formal proof yet.
\begin{conjecture}\label{conj}
Let $ n $ and $ k $ be two integers such that $ 2\leq k\leq n $ and consider the function $ F_k(n) $ defined in Eq.\ (\ref{eq:upp_rt_h}). Then it holds that $F_{k}(n)= \min_{2\leq h\leq k}\lbrace B_{h}(n)+ \tilde U^k_{h}(n)\rbrace= B_{2}(n)+ \tilde U^k_{2}(n) $, which implies that
\begin{equation}\label{eq:conj}
rt_k(n)\leq 1+ \tilde U^k_{2}(n).
\end{equation}
\end{conjecture}
Conjecture \ref{conj} suggests that the knowledge of the function $  \tilde U^k_{h}(n)$ for $ h=2 $ is enough to establish a (better) upper bound on $ rt_k(n) $, for every $ 2\leq k\leq n $. Since $\tilde U^k_{2}(n)  $ is an upper bound on $ O^k_2(n) $ (defined in Eq.\ (\ref{eq:=Okn})), Conjecture \ref{conj} also suggests that \emph{summing up} columns (or rows) of weight two in a matrix should be the quickest way to obtain a column (or row) of weight $\geq  k $, for any $ k $. For further (and future) improvements on the upper bound on $ rt_k(n) $, we could then think to better exploit the usage of columns of weight two.

 Numerical simulations also suggest that the two bounds $ B_k(n) $ and $ F_k(n) $ should coincide for $ n $ big enough (and fixed $ k $); this is formalized in the following conjecture.

\begin{conjecture}\label{conj2}
Let $ k \in \mathbb{N}$ and $ n_k=2k^2-8k+12 $. Then:
\begin{enumerate}
\item if $ 2\leq k\leq 6 $, for all $ n>k^2 $ it holds that $ F_k(n)=B_k(n) $;
\item if $ k> 6 $, for all $ n>n_k $ it holds that $ F_k(n)=B_k(n) $.
\end{enumerate}
\end{conjecture}
Notice that $ n_k>k^2 $ for all $ k>6 $.
Conjecture \ref{conj2}, if true, states that our two upper bounds on the $ k $-RT coincide when $ k $ is fixed and $ n $ is big enough. Equivalently, our two upper bounds on the $ k $-RT coincide when $ n $ is fixed and $ k $ is small enough. 
This would imply that the linear upper bound on $ rt_k(n) $ for fixed $ k\leq \sqrt{n} $ in Eq.\ (\ref{second rec}) is not improved by our new techinque. On the other hand, for values of $ n $ and $ k $ not fulfilling requirements 1.\ and 2.\ of Conjecture \ref{conj2}, numerical results seem to suggest that $ F_k(n)<B_k(n) $, that is our new techinque strictly improves the bound $ B_k(n) $. This will be shown in detail in the next section, in particular in Figures \ref{fig:comparisonB_k_fixed} and \ref{fig:comparisonB_n_fixed}.
Finally, Figure \ref{fig:nk} shows for each fixed $ 6<k\leq 50 $, the smallest value of $ n<5000 $ found such that $ F_k(n')=B_k(n') $ for all $ n' $ s.t.\ $ n<n'<5000 $; it clearly supports Conjecture \ref{conj2}.
\begin{figure}
\centering
\includegraphics[scale=0.16]{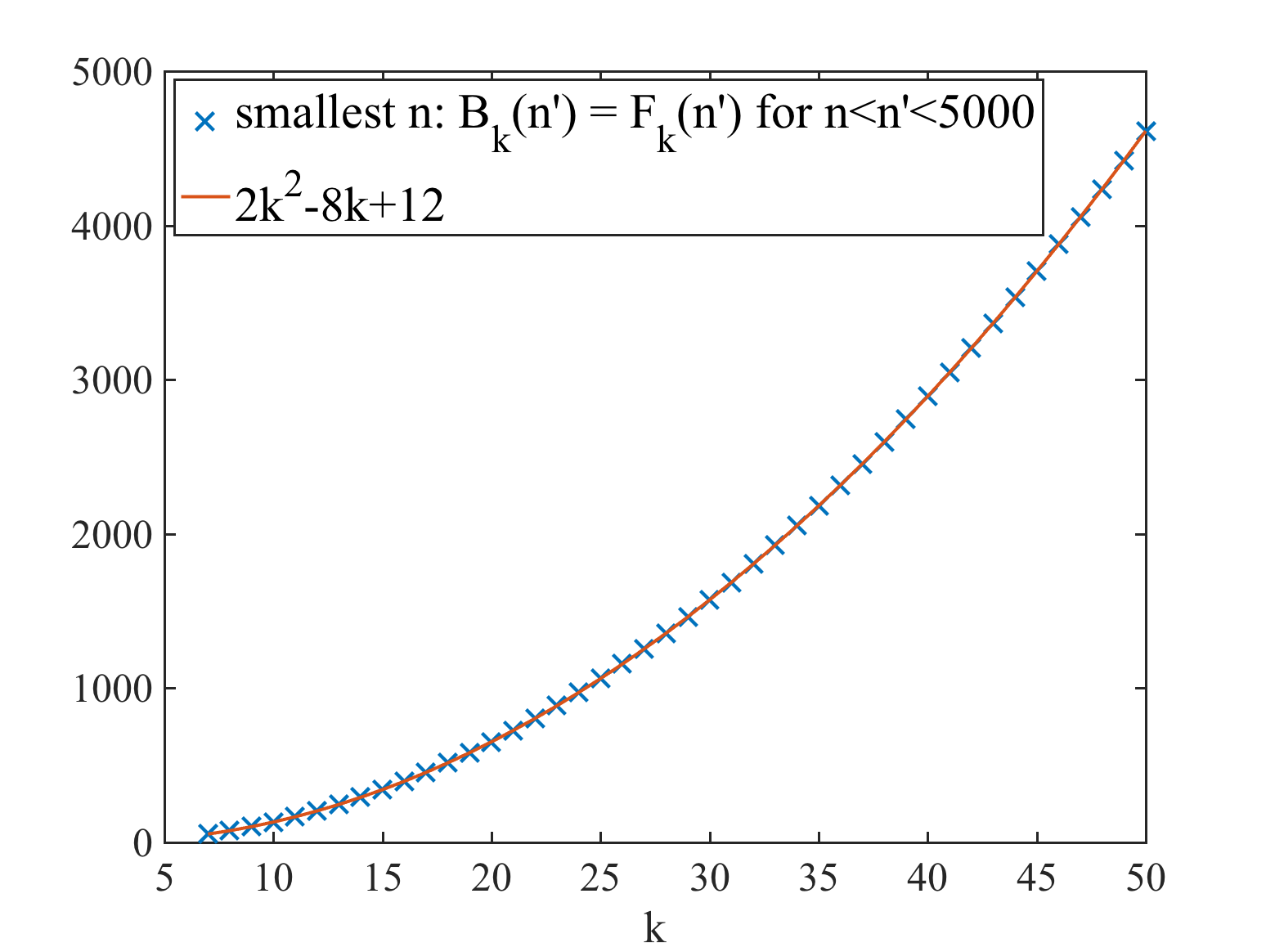}
\caption{Comparison between the function $ 2k^2-8k+12 $ and the values of the smallest $ n $ for which resulted that $ F_k(n')=B_k(n') $ for all $ n' $ such that $ n<n'<5000 $.}
\label{fig:nk}
\end{figure}

\section{Numerical results for the upper bound $ F_h(n) $}\label{sec:num_F}
In this section, we first present the analogue of Figures \ref{fig:test}, \ref{fig:test2} and \ref{fig:test3} for the new upper bound $ F_k(n) $, and then we show some graphs picturing the behavior of $ F_k(n) $, first for fixed $ k $ and then for fixed $ n $.
Figure \ref{fig:test_new} compares $ F_k(n) $ with $ B_k(n) $ and with the real $ k $-RT of the primitive sets $\mathcal{M}_{CPR}$ and $\mathcal{M}_{K}$ (see Section \ref{sec:num_B}). We can see that for $k\leq \sqrt{n}$, the upper bound $ F_k(n) $ coincides with $ B_k(n) $ and it is fairly close to the actual value of $rt_k(\mathcal{M})$, while for $k> \sqrt{n}$, $ F_k(n) $ strictly improves on $ B_k(n) $ in almost all cases.
\begin{figure}
\centering
\subfigure[][$n=4$, $\mathcal{M} = \mathcal{M}_{CPR}$]{\includegraphics[scale=0.15]{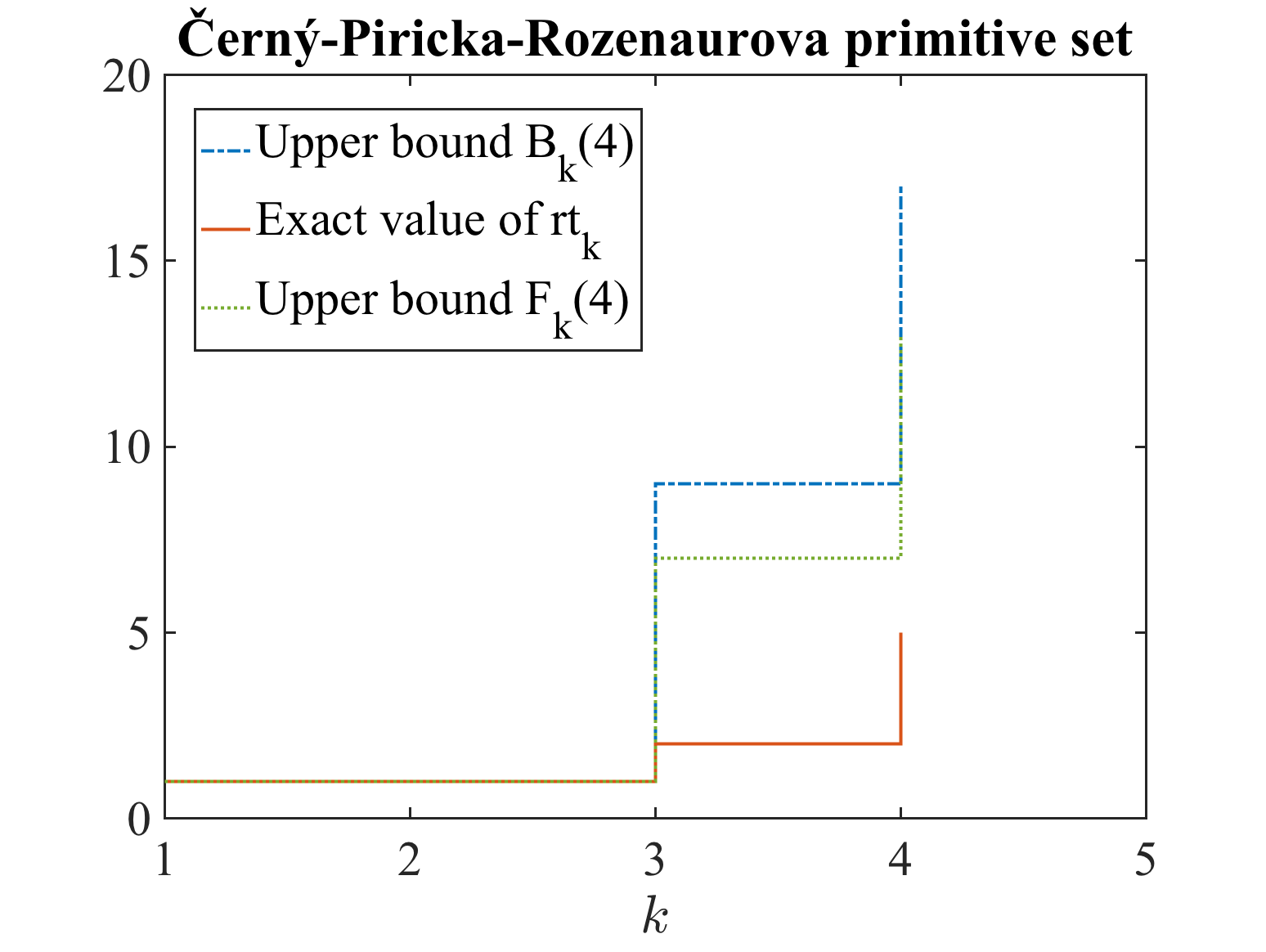}}\subfigure[][$n=6$, $\mathcal{M} = \mathcal{M}_{K}$]{\includegraphics[scale=0.15]{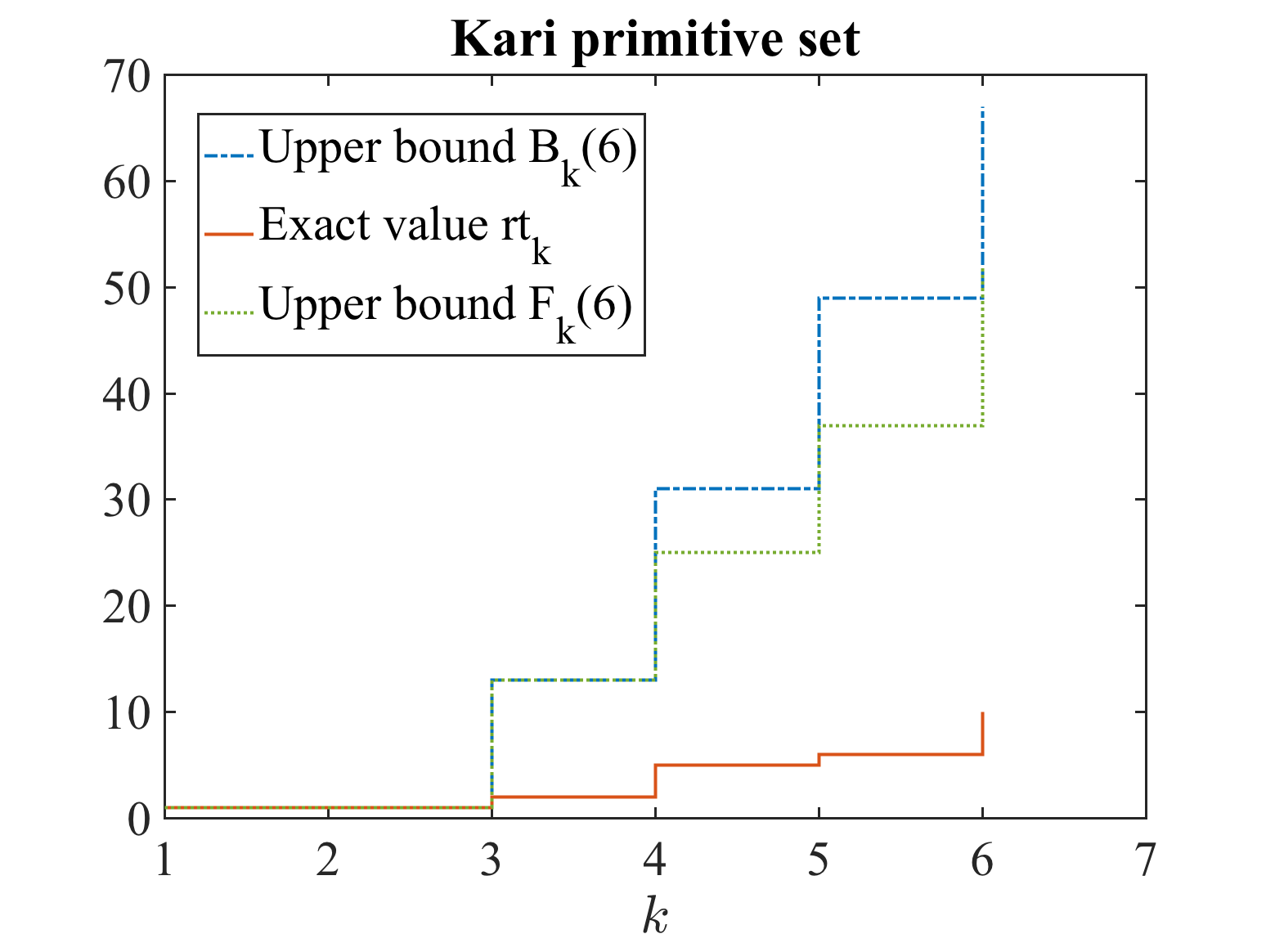}}
\caption{Comparison between the upper bounds $ F_k(n) $ and $ B_k(n)$, valid for all primitive NZ sets, and $rt_k(\mathcal{M})$ for (a) $\mathcal{M} = \mathcal{M}_{CPR}$ and (b) $ \mathcal{M} = \mathcal{M}_{K}$.}
\label{fig:test_new}
\end{figure}
Figure \ref{fig:test2_new} compares the upper bounds $ F_k(n) $ and $B_k(n) $, and the results of the Eppstein heuristic on the $k$-RT of the primitive sets $\mathcal{M}_{C_{10}}$, $\mathcal{M}_{C_{15}}$, $\mathcal{M}_{C_{20}}$ and $\mathcal{M}_{C_{25}}$ (see Section \ref{sec:num_B}). It can be noticed that for $ k $ big enough, $ F_k(n) $ substantially improves on $B_k(n) $.
\begin{figure}[h!]
\centering 
\subfigure[][$\mathcal{M} = \mathcal{M}_{C_{10}}$]{\includegraphics[scale=0.18]{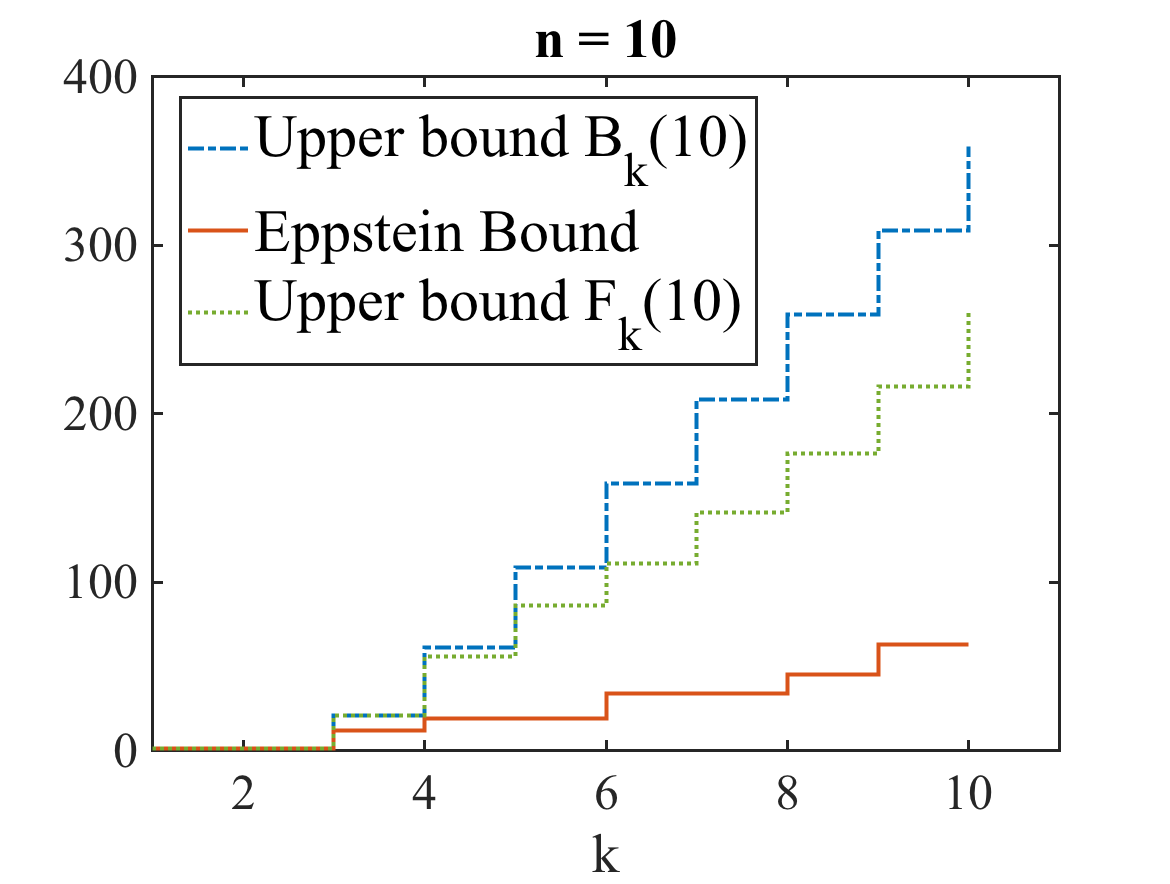}}\subfigure[][$\mathcal{M} = \mathcal{M}_{C_{15}}$]{\includegraphics[scale=0.18]{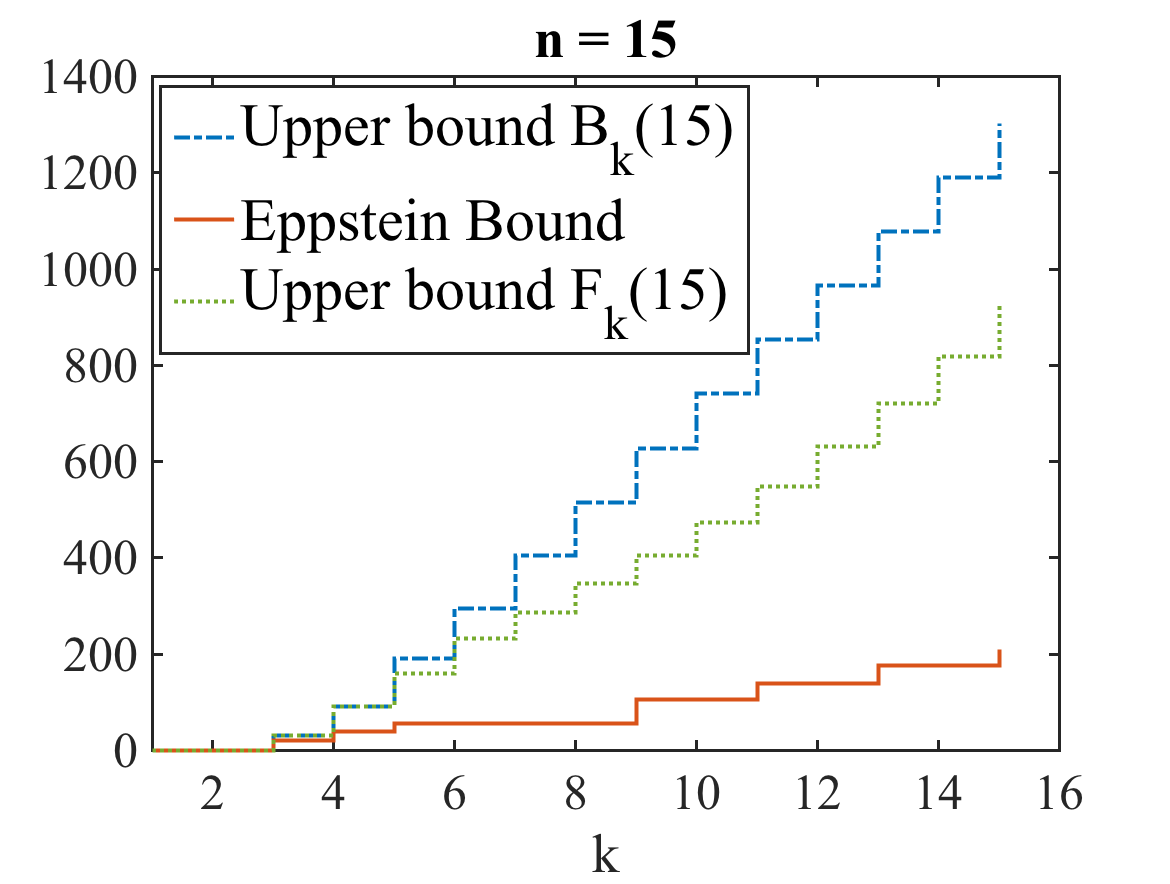}}\\
\subfigure[][$\mathcal{M} = \mathcal{M}_{C_{20}}$]{\includegraphics[scale=0.18]{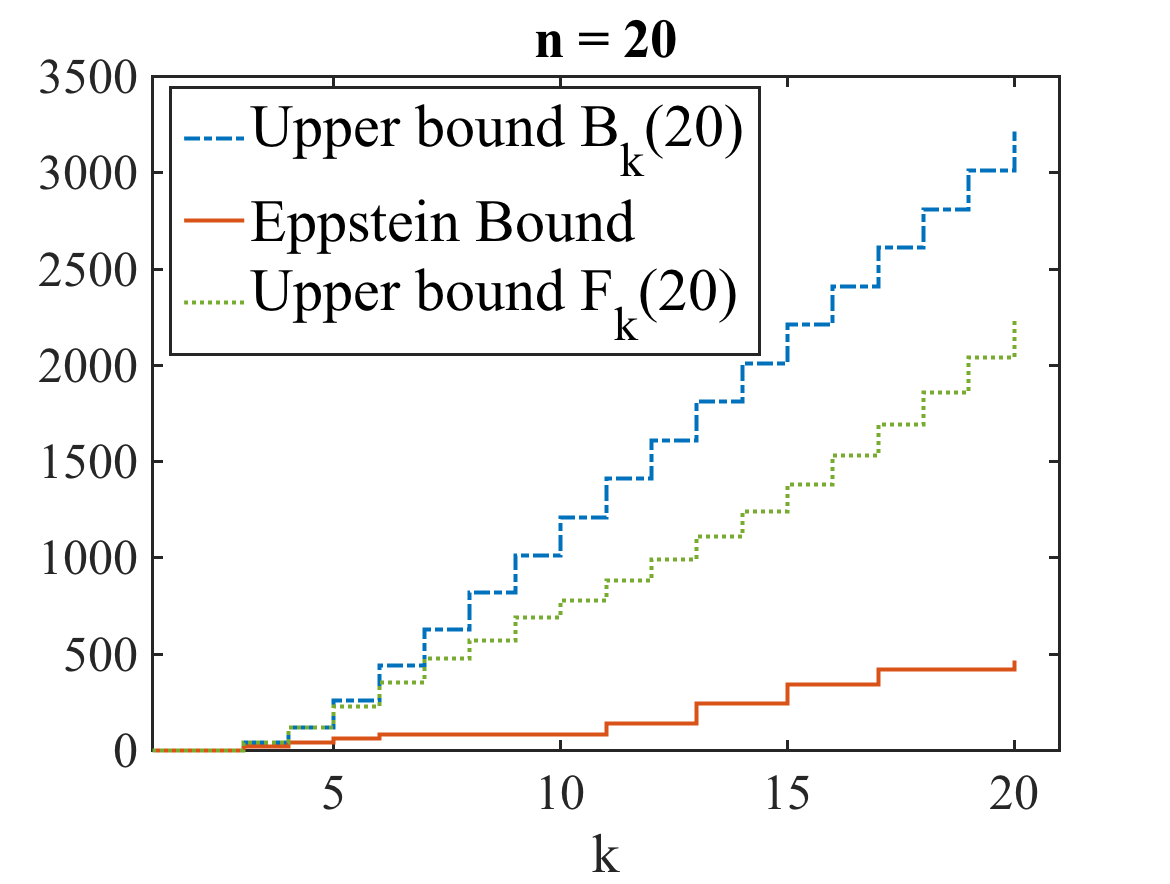}}\subfigure[][$\mathcal{M} = \mathcal{M}_{C_{25}}$]{\includegraphics[scale=0.18]{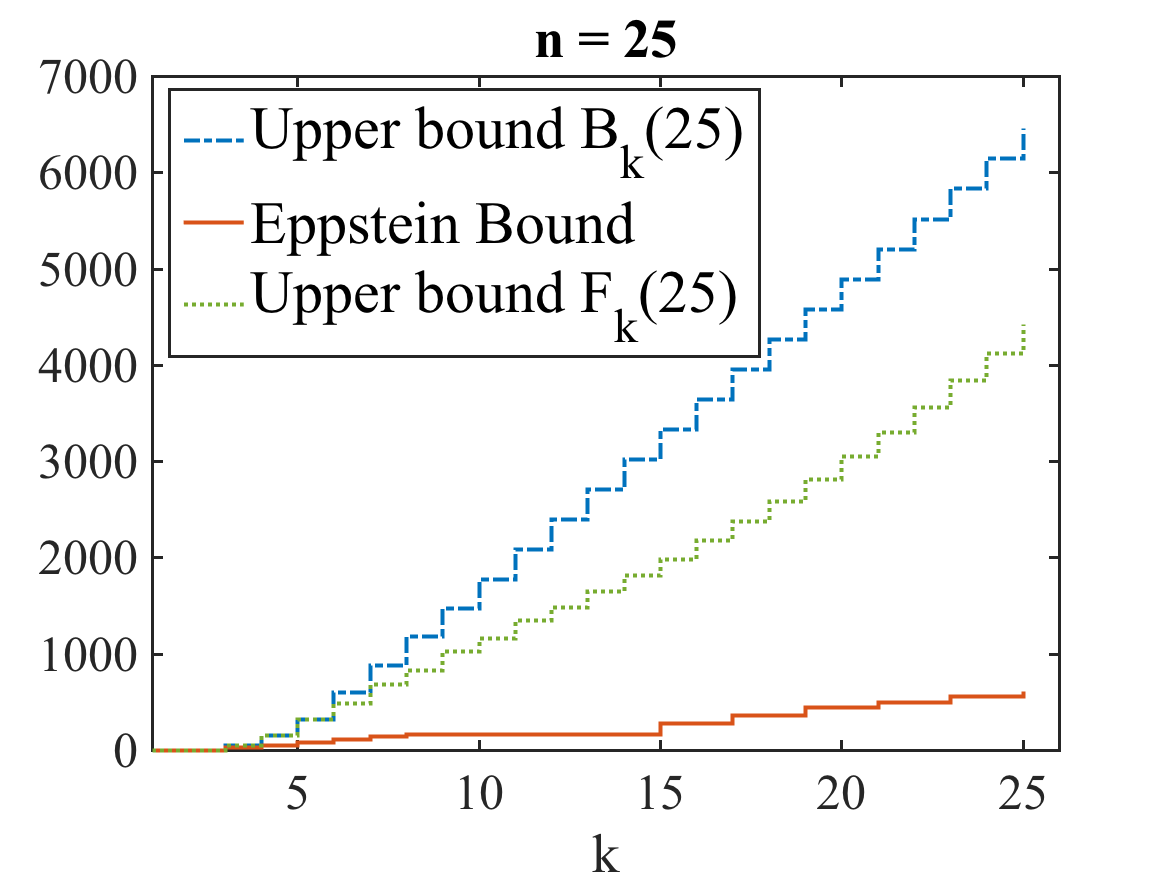}}  
\caption{Comparison between the bounds $ F_k(n) $ and $B_k(n)$, and the Eppstein approximation of $rt_k(\mathcal{M})$, for (a) $\mathcal{M} = \mathcal{M}_{C_{10}}$, (b) $\mathcal{M} = \mathcal{M}_{C_{15}}$, (c) $\mathcal{M} = \mathcal{M}_{C_{20}}$, (d) $\mathcal{M} = \mathcal{M}_{C_{25}}$.}
\label{fig:test2_new}
\vspace{-0.3cm}
\end{figure}
Regarding Figure \ref{fig:test3}, since $ k$ is fixed equal to $4 $ and $ n $ ranges from $ 21$ to $30 $, in view of Conjecture \ref{conj2} (and confirmed by numerical experiments) the values of $ F_k(n) $ coincide with the values of $ B_k(n) $.  

Figure \ref{fig:comparisonB_k_fixed} compares the new upper bound $ F_k(n) $ on the $ k $-RT with the bound $ B_k(n) $ obtained in Section \ref{sec:impro} for fixed values of $ k $ ($ k=10,20,30,40,50,100 $). The pictures show the values of $ n $ for which $ F_k(n) $ is a real improvement of $ B_k(n) $, as for $ n $ big enough the two bounds seem to coincide (see also Conjecture \ref{conj2}). Furthermore for every $ k $ and $ n $ we obtained that the minimum of (\ref{eq:upp_rt_h}) is reached at $ h=2 $, that is $F_{k}(n)= B_{2}(n)+ \tilde U^k_{2}(n) $, thus supporting Conjecture \ref{conj}.

Figure \ref{fig:comparisonB_n_fixed} compares $ F_k(n) $ with $ B_k(n) $ for fixed values of $ n $ ($ n=10,50,100,200,500,1000 $) and $ k $ ranging from $ 2 $ to $ n $. We again observe that for small values of $ k $ the two bounds coincide (see also Conjecture \ref{conj2}), while for bigger values of $ k $ the bound $ F_k(n) $ is a substantial improvement of $ B_k(n)  $. Again, for every $ k $ and $ n $ we obtained that the minimum of (\ref{eq:upp_rt_h}) is reached at $ h=2 $, that is $F_{k}(n)= B_{2}(n)+ \tilde U^k_{2}(n) $, supporting Conjecture \ref{conj}.

Finally, Figure \ref{fig:nRT_comparisons} compares the upper bounds $ F_k(n) $ and $ B_k(n) $ for $ k=n $, that is the upper bound on the length of the smallest product having a column or a row with all positive entries. Since by Proposition \ref{prop:krt_prim_aut} any generic upper bound on the reset threshold of a synchronizing automaton is an upper bound on $ rt_n(n) $, in Figure \ref{fig:nRT_comparisons} we also represent the function $(15617 n^3 + 7500 n^2 + 9375 n - 31250)/93750$, which is the upper bound on the reset threshold of a synchronizing automaton on $ n $ states found by Szyku{\l}a in \cite{Szykula}. Szyku{\l}a's upper bound has been recently improved by Shitov \cite{Shitov}, who proved that we can upper bound the reset threshold of a synchronizing automaton by a function $ \alpha n^3+o(n^3) $ with $ \alpha\leq 0.1654 $ (while in Szykula's upper bound $ \alpha \approx 0.1664 $). We decided to picture Szyku{\l}a's bound because it has a precise analytical expression, while Shitov's one does not, and because the difference between the two bounds is negligible with respect to our purposes. Indeed Figure \ref{fig:nRT_comparisons} shows that up to now, our techniques fall short of improving the upper bound on the $ n $-RT, as more efficient bounds are already known. 

\begin{figure}[p]
\centering 
\subfigure[][$ k=10 $]{\includegraphics[scale=0.16]{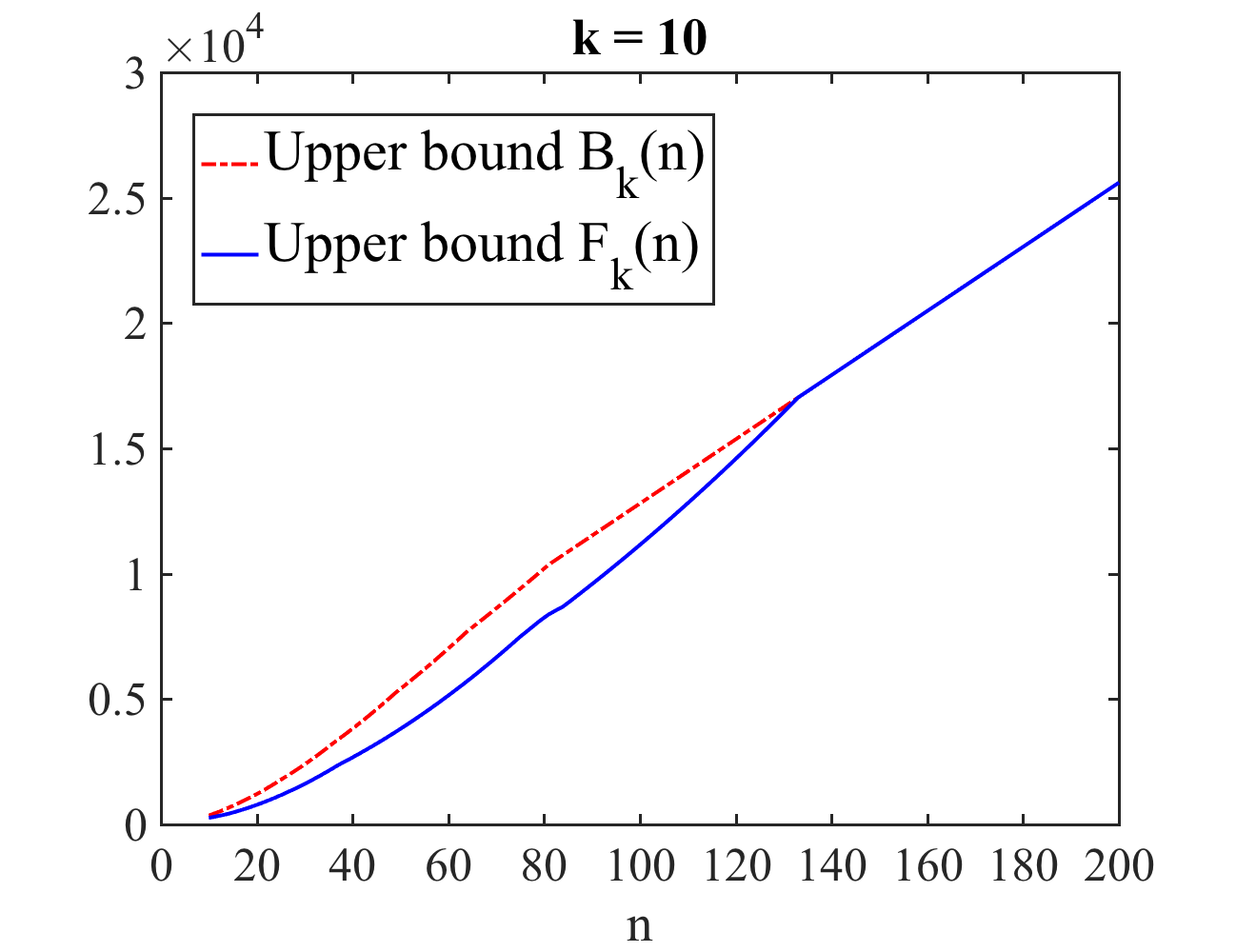}}
  \subfigure[][$ k=20 $]{\includegraphics[scale=0.16]{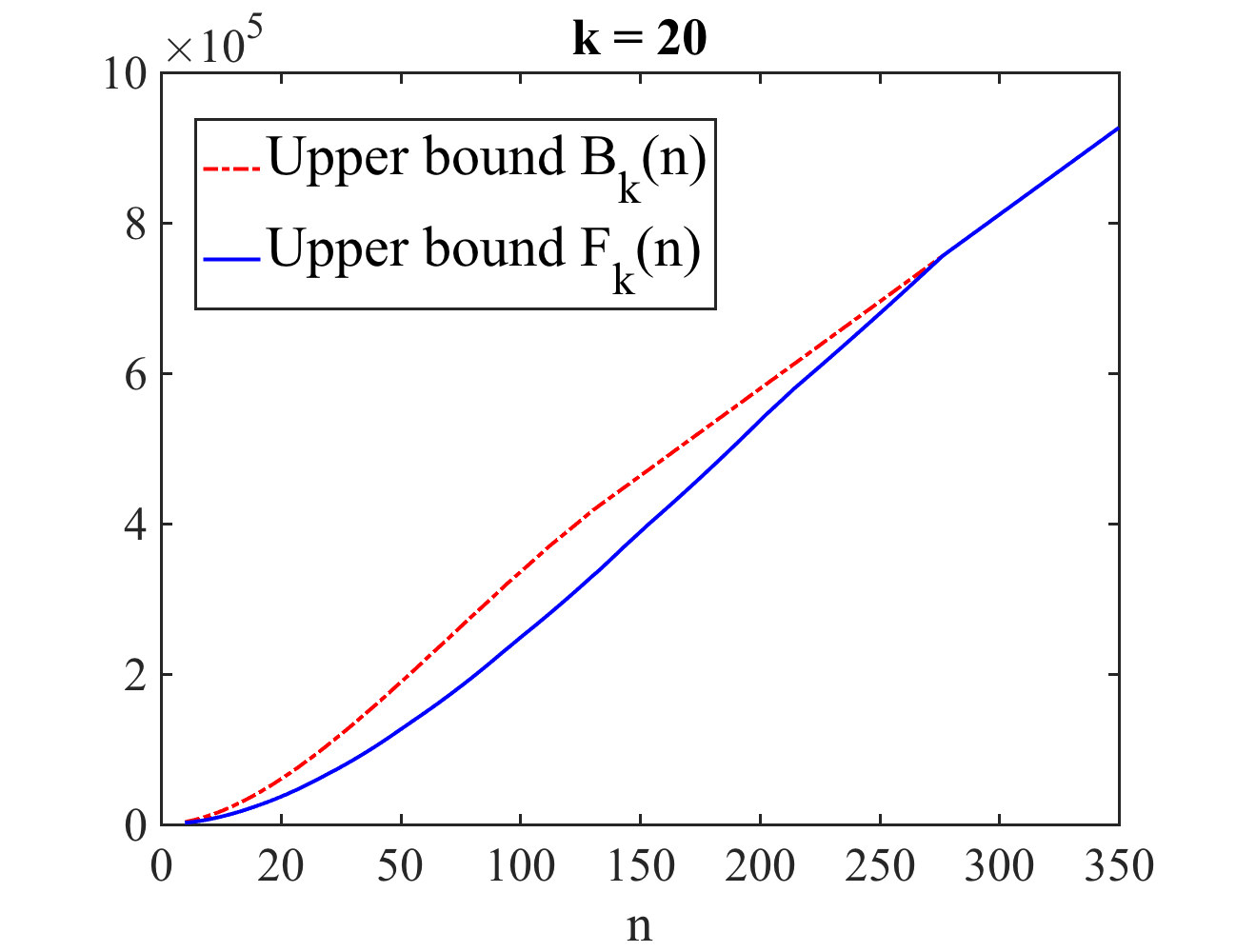}}\\
  \subfigure[][$ k=30 $]{\includegraphics[scale=0.16]{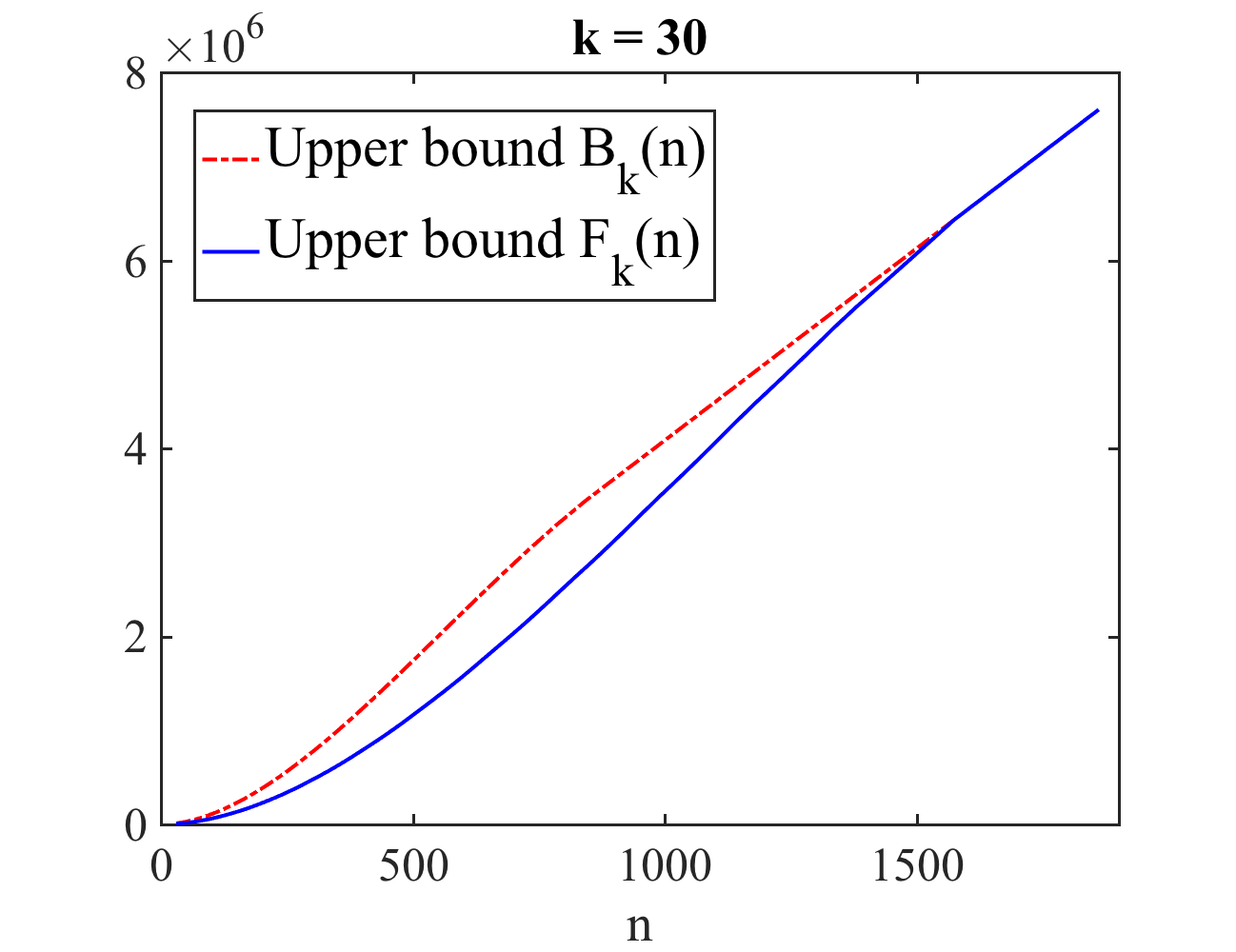}}
  \subfigure[][$ k=40 $]{\includegraphics[scale=0.16]{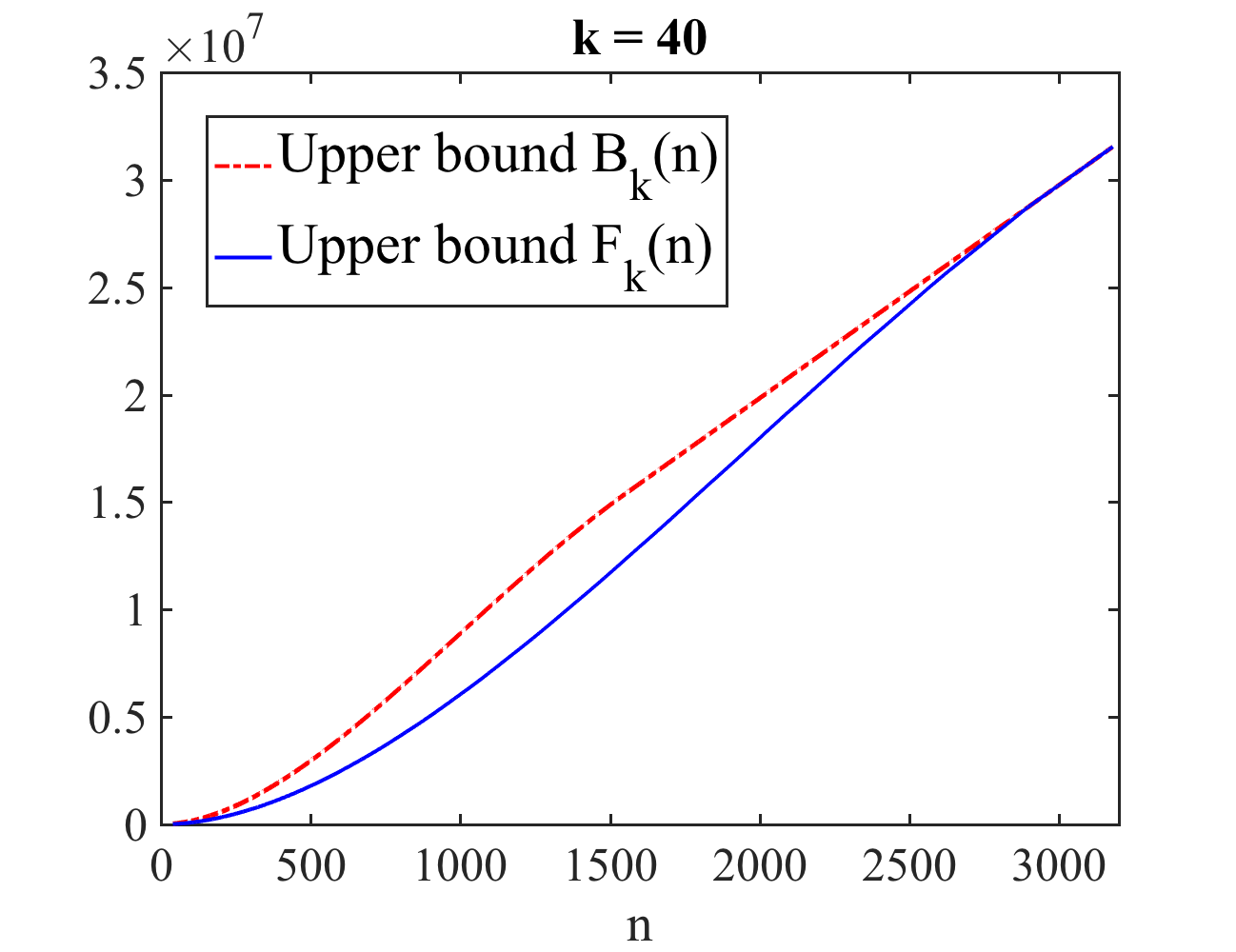}}\\
  \subfigure[][$ k=50 $]{\includegraphics[scale=0.16]{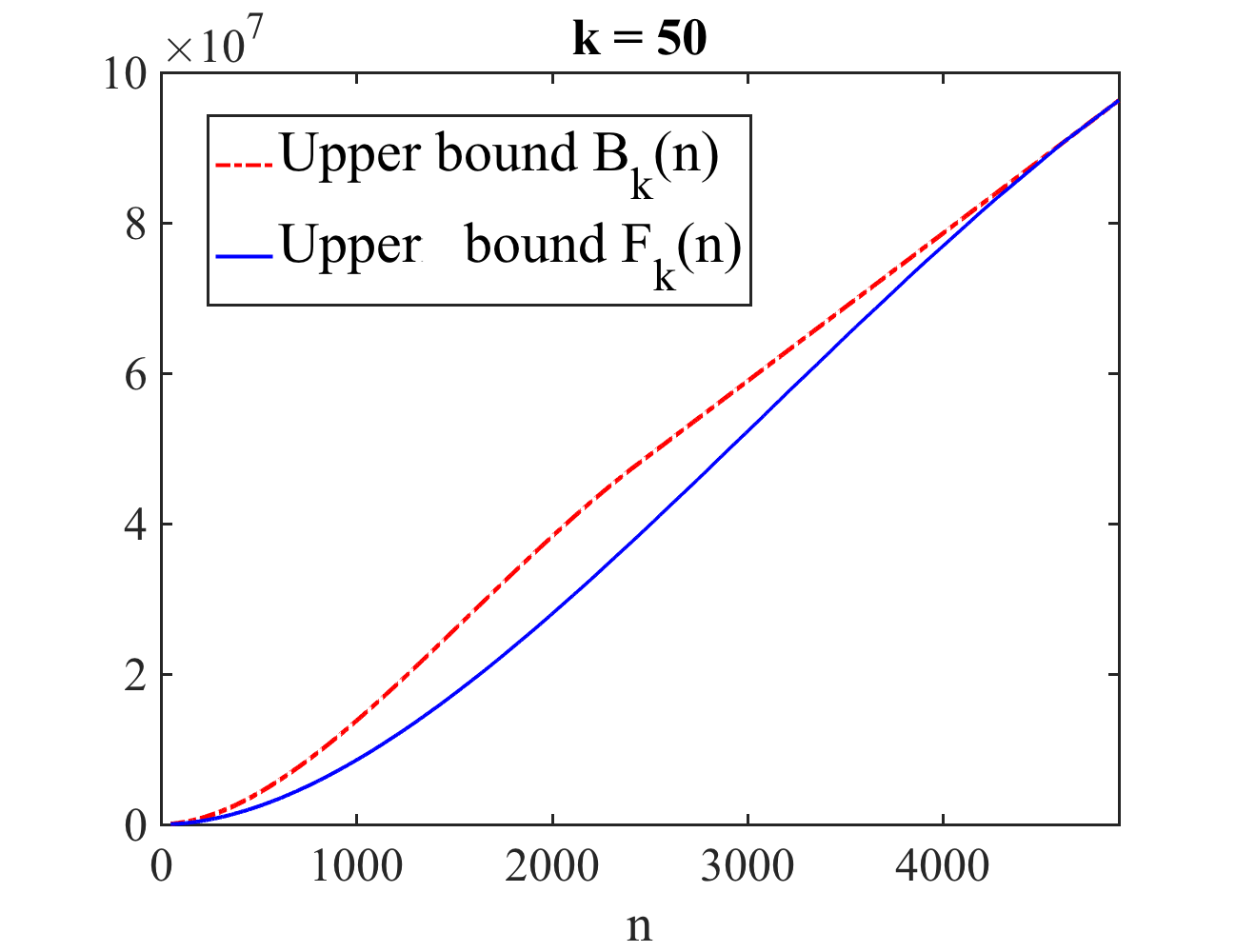}}
  \subfigure[][$ k=100 $]{\includegraphics[scale=0.16]{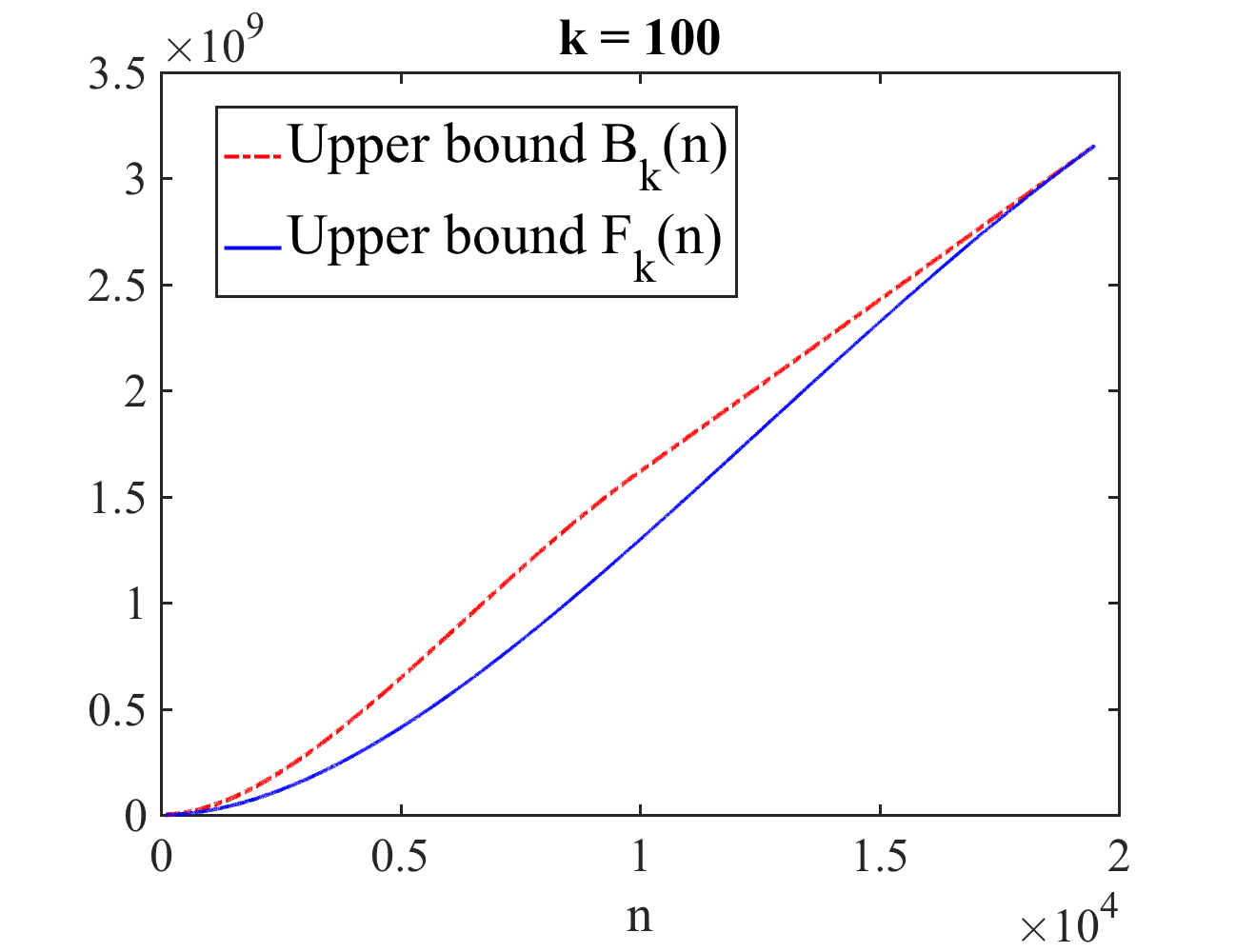}}
\caption{Comparison between the bounds $ F_k(n) $ and $ B_k(n)$ for fixed values of $ k $.}
\label{fig:comparisonB_k_fixed}
\end{figure}

\begin{figure}[p]
\centering 
\subfigure[][$ n=10 $]{\includegraphics[scale=0.145]{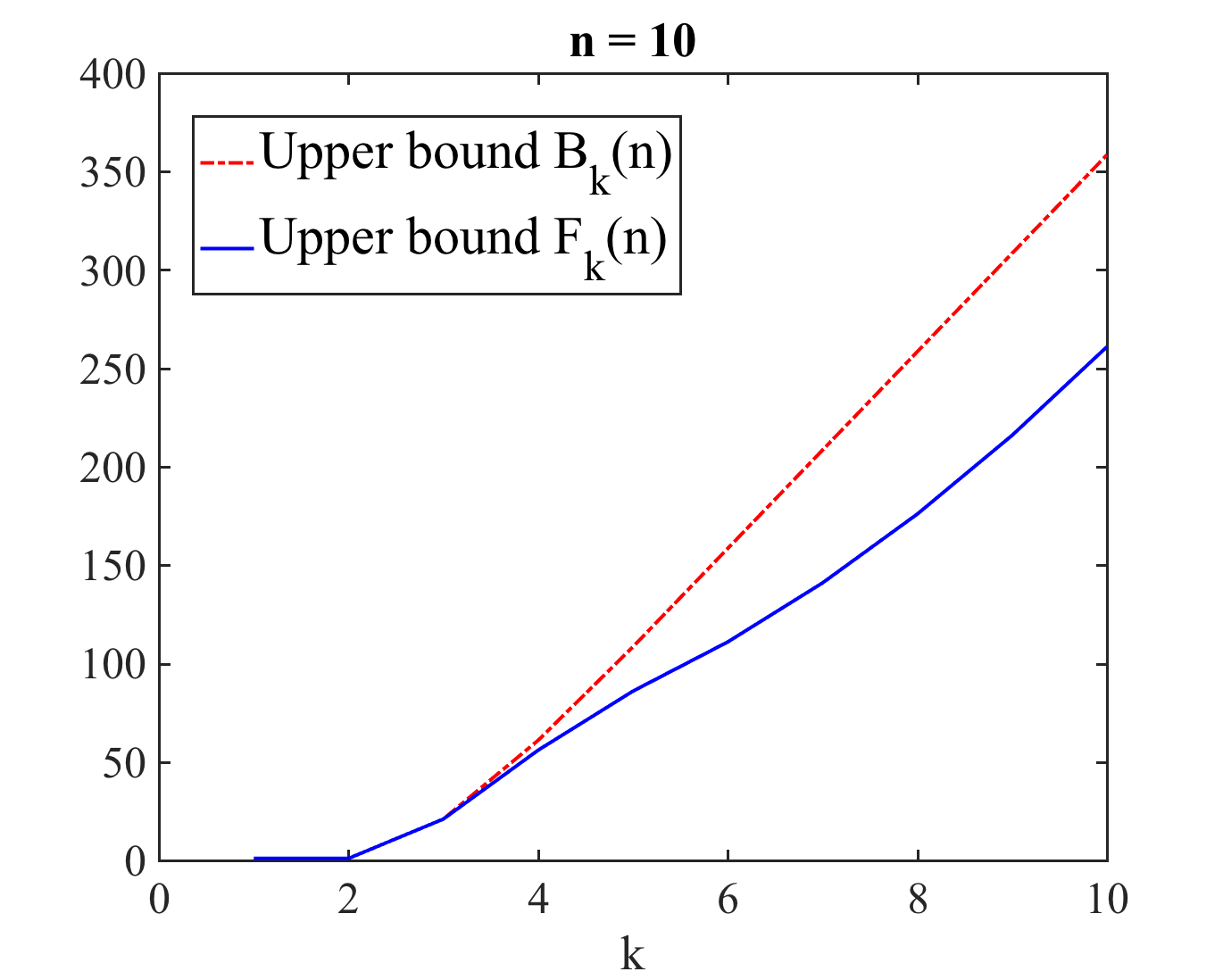}}
  \subfigure[][$ n=50 $]{\includegraphics[scale=0.145]{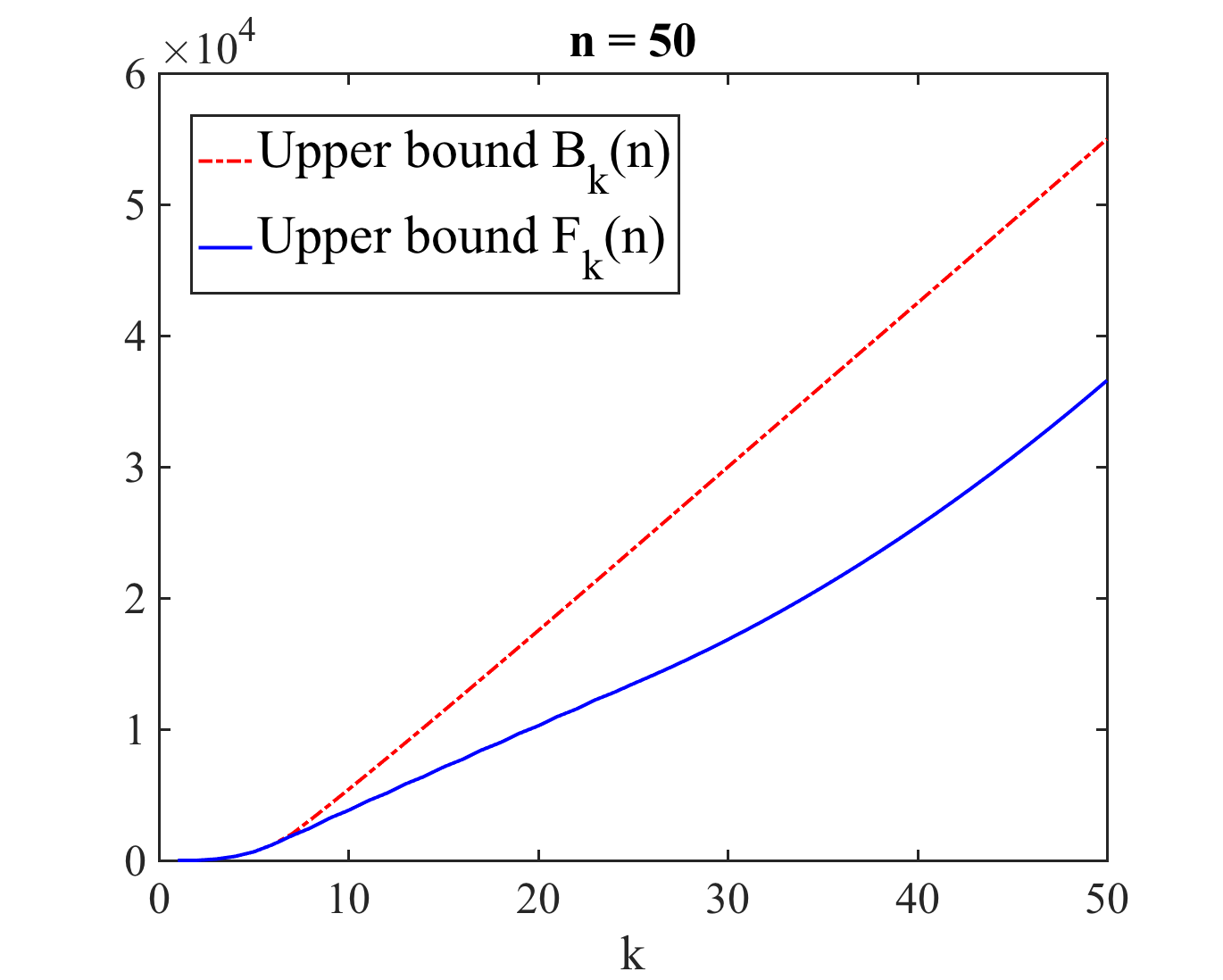}}\\
  \subfigure[][$ n=100 $]{\includegraphics[scale=0.145]{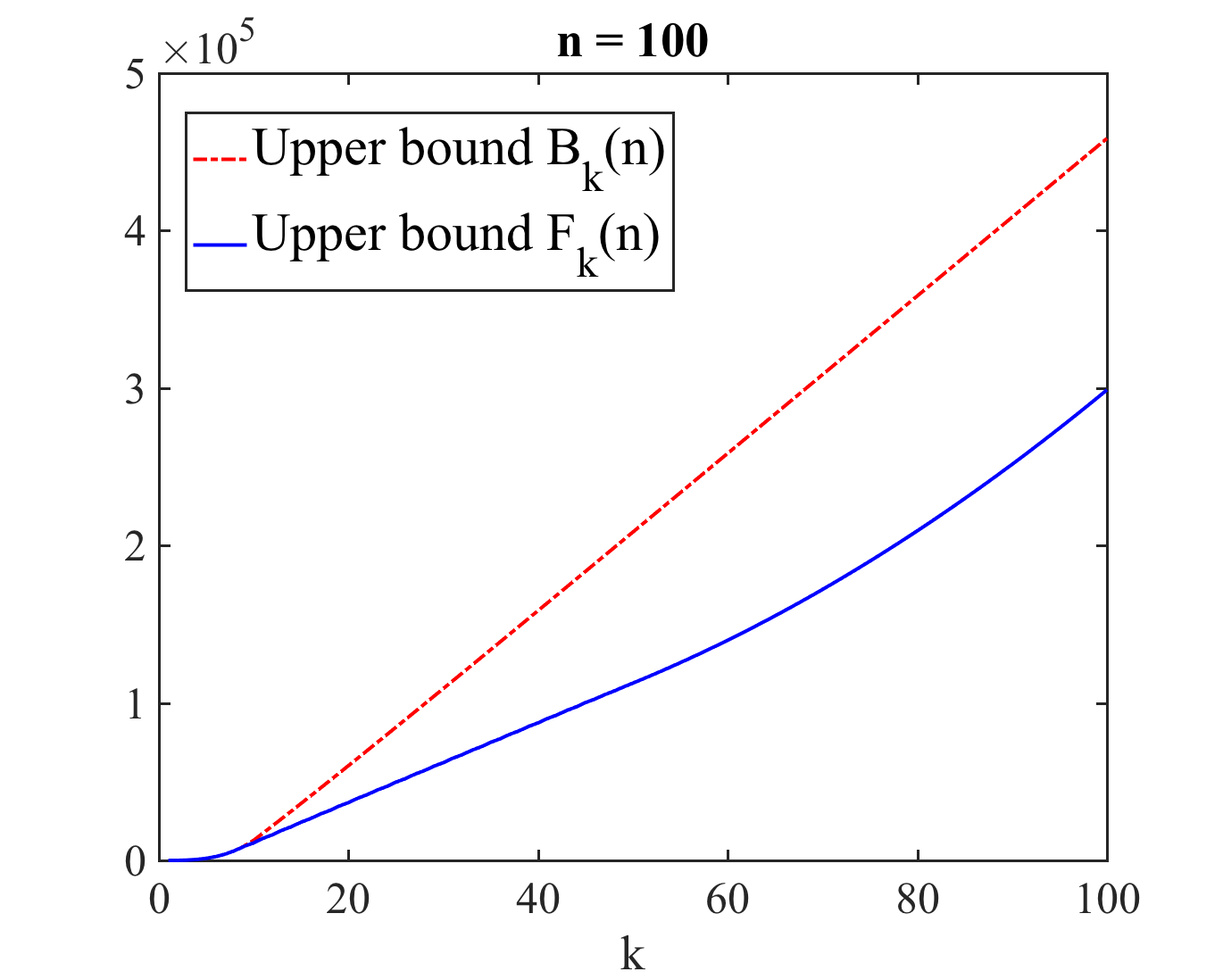}}
  \subfigure[][$ n=200 $]{\includegraphics[scale=0.145]{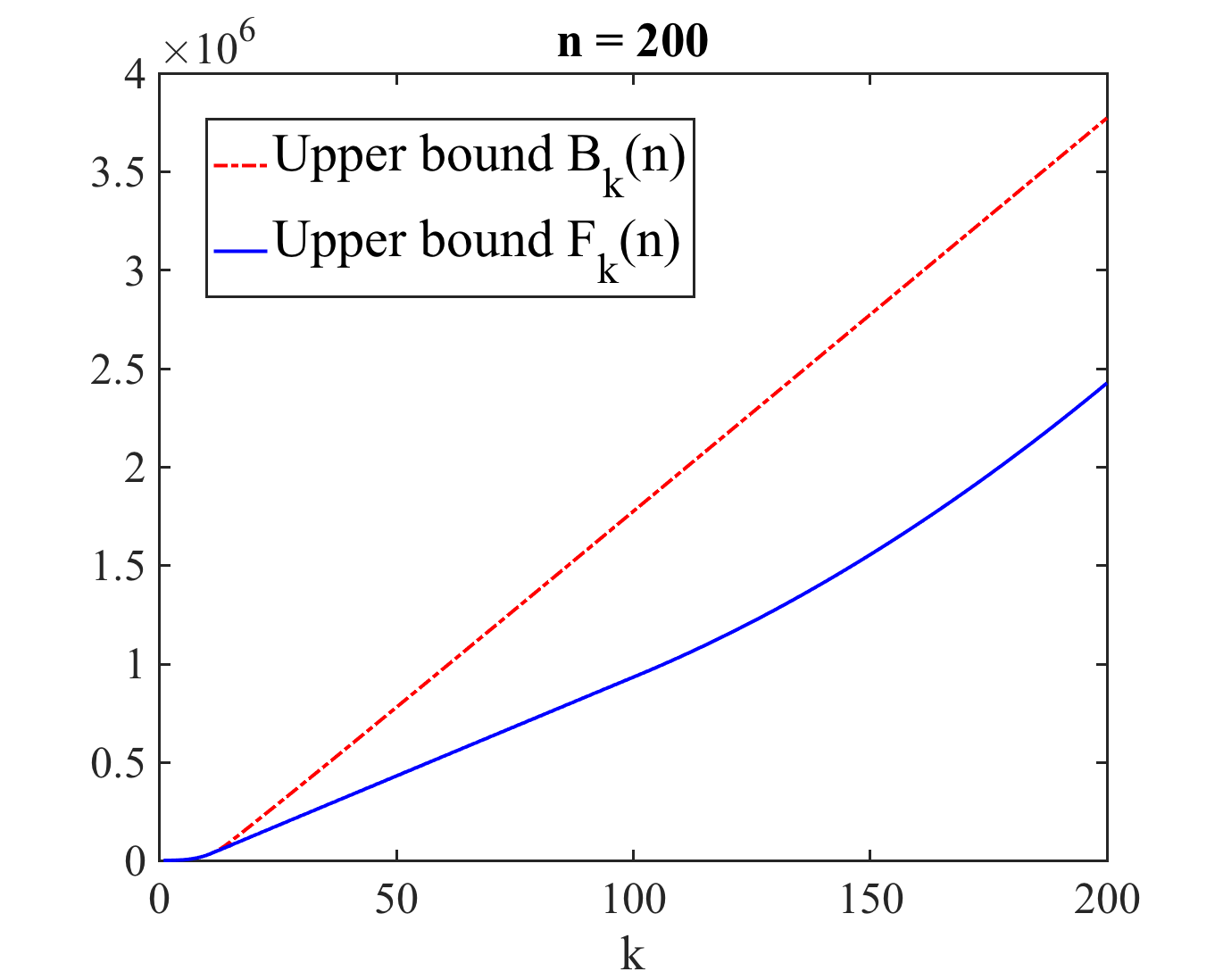}}\\
  \subfigure[][$ n=500 $]{\includegraphics[scale=0.145]{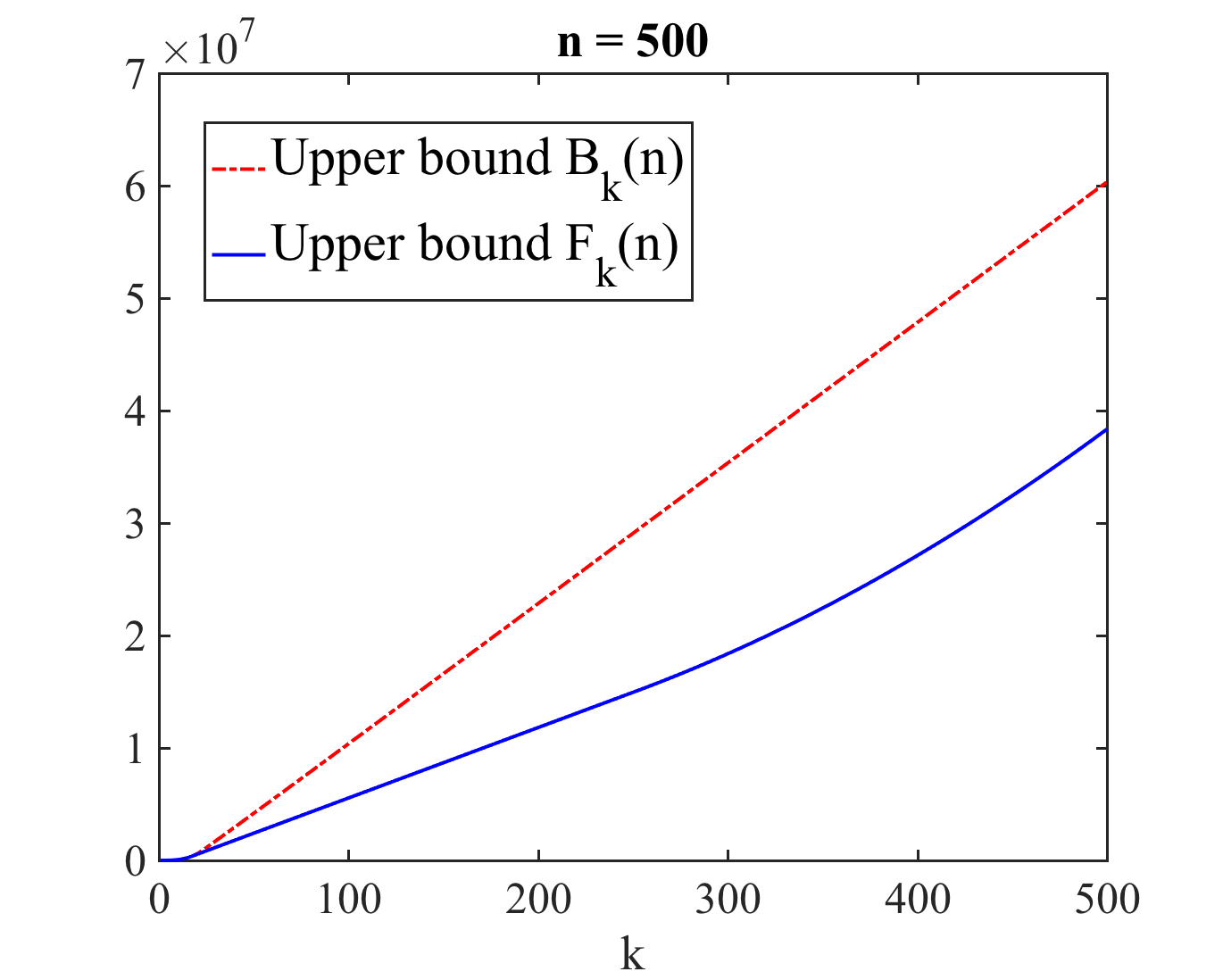}}
  \subfigure[][$ n=1000 $]{\includegraphics[scale=0.145]{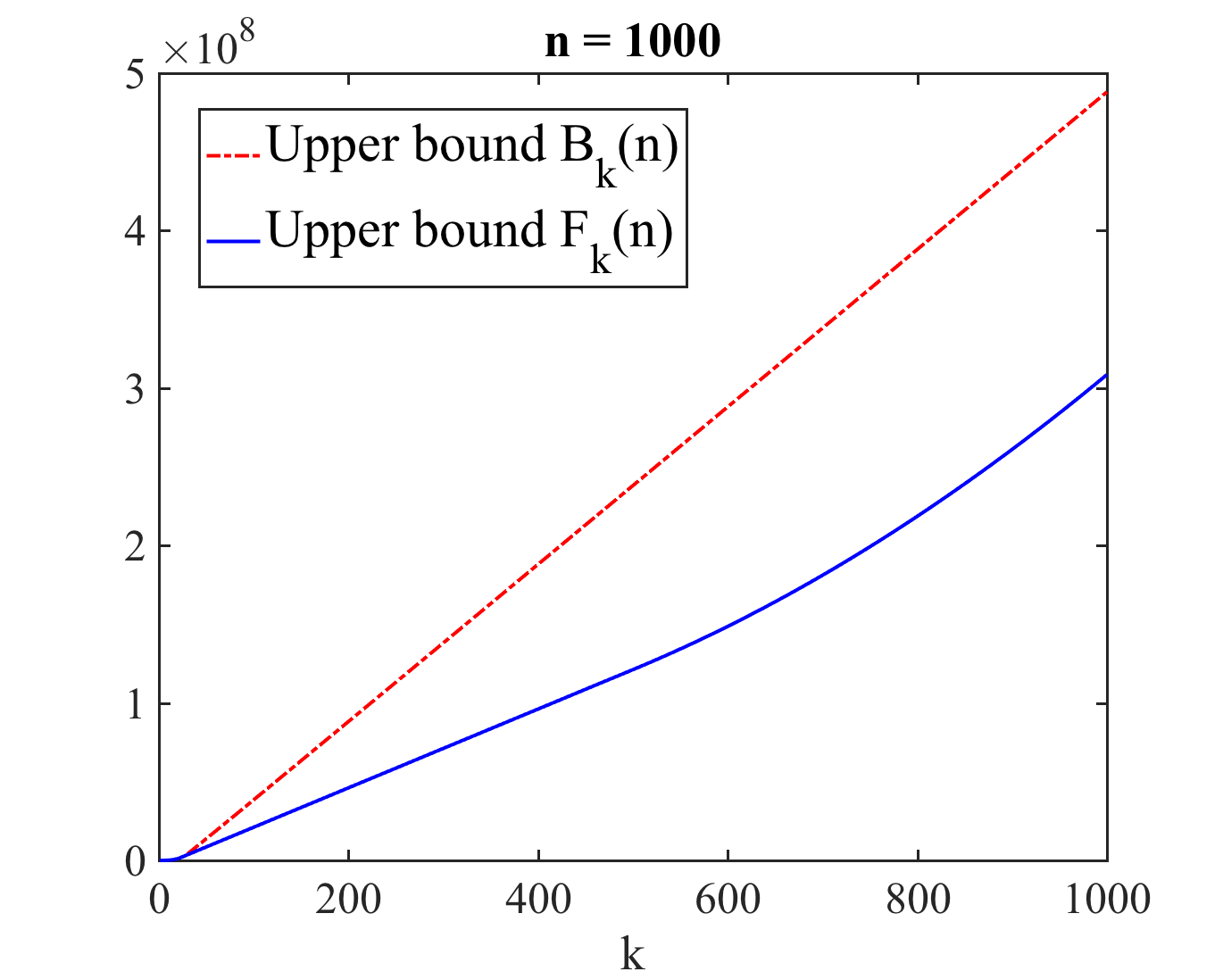}}
\caption{Comparison between the bounds $ F_k(n) $ and $B_k(n)$ for fixed values of $ n $.}
\label{fig:comparisonB_n_fixed}
\end{figure}

\begin{figure}[h!]
\centering
\includegraphics[scale=0.2]{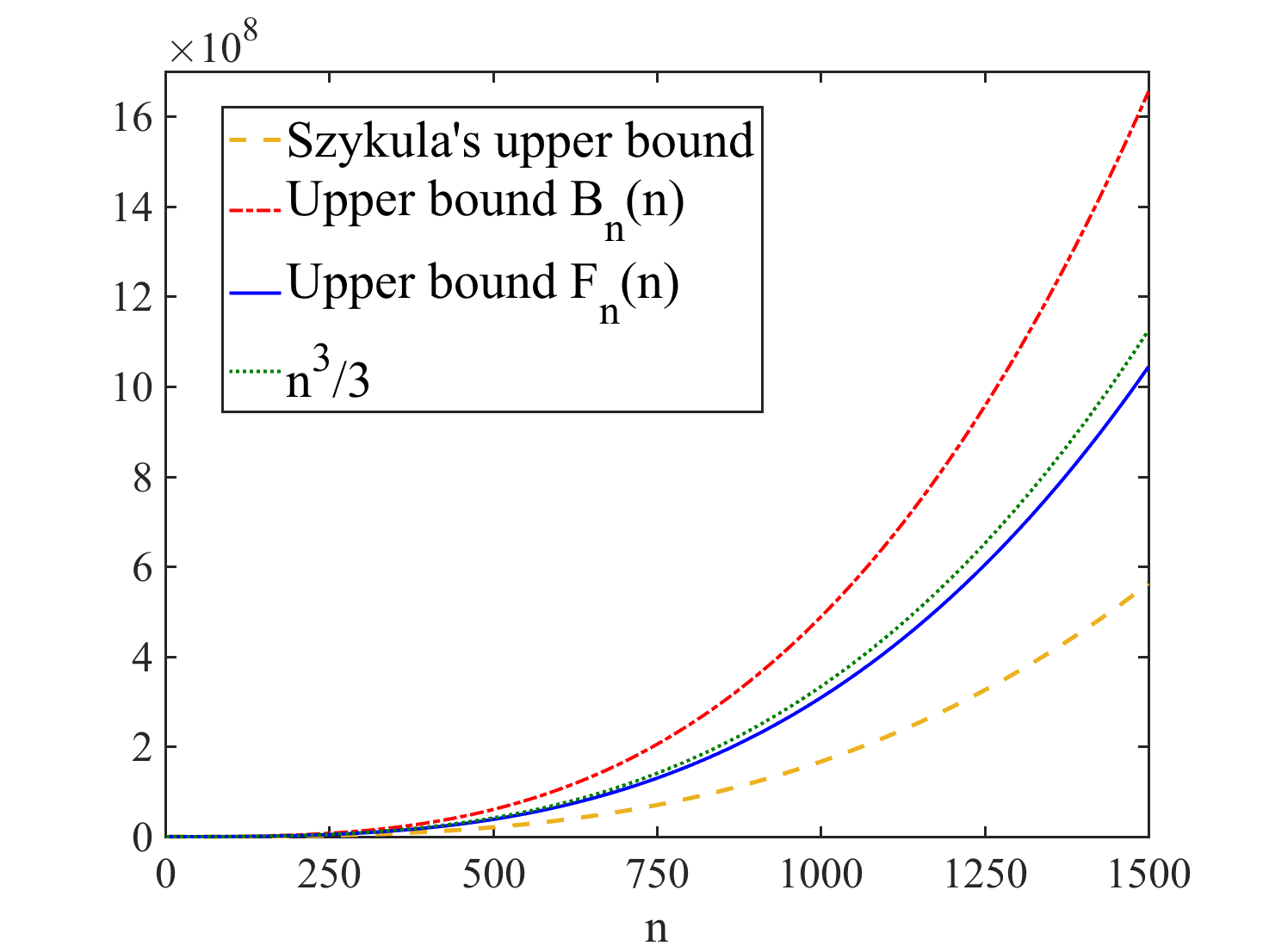}
\caption{Comparison between the two upper bounds $ F_n(n) $, $ B_n(n) $, and Szyku{\l}a's upper bound. The function $ n^3/3 $ has been pictured for reference.}
\label{fig:nRT_comparisons}
\end{figure}

\section{Conclusions}\label{sec:conclusions}
In this paper we have shown that we can upper bound the length of the shortest product of a primitive NZ set $ \mathcal{M} $ having a column or a row with $ k $ positive entries by a linear function of the matrix size $ n $, for any constant $ k\leq \sqrt{n} $. We have called this length the $ k $\emph{-rendezvous time} ($ k $-RT) of the set $ \mathcal{M} $, and we have shown that the same linear upper bound holds for $ \min \bigl\lbrace rt_k\bigl(Aut(\mathcal{M})\bigr), rt_k\bigl(Aut(\mathcal{M}^{\top})\bigr)\bigr\rbrace $, where $ Aut(\mathcal{M}) $ and $ Aut(\mathcal{M}^{\top}) $ are the synchronizing automata defined in Definition \ref{def:assoc_autom}. We have also showed that our technique cannot be improved as it is, because it already takes into account the worst cases. We have then presented a new strategy to obtain a better upper bound on $ rt_k(n) $, which takes into account the weights of the columns (or rows) that we are summing up to obtain a column (or row) of higher weight; numerical results show that this new upper bound significantly improves the previous one when $ n $ is not too large with respect to $ k $. The notion of $ k $-RT comes as an extension of a similar notion for synchronizing automata introduced in \cite{Gonze2015}. For automata, the problem whether there exists a linear upper bound on the $ k $-RT for small $ k $ is still open, as the only nontrivial result on the $ k $-RT that appears in the literature proves a quadratic upper bound on the $ 3 $-RT \cite{Gonze2015}. We believe that our result, as well as the new technique developed in Section \ref{sec:improv_ub}, could help in shedding light to this problem and possibly to the \v{C}ern\'{y} conjecture, in view of the connection between synchronizing automata and primitive NZ sets established by Theorem \ref{thm:autom_matrix}.

\bibliography{biblOnprobprim}

\bibliographystyle{fundam}



\end{document}